\documentclass[11pt]{article}

\usepackage{amsthm}
\usepackage{amsmath}
\usepackage{amsfonts}
\usepackage{mathrsfs}
\usepackage{dsfont}

\usepackage{amssymb}
\usepackage{enumerate}
\usepackage{graphicx}
\usepackage{mathrsfs}

\usepackage{a4}
\usepackage{url}
\usepackage{algorithm}
\usepackage{algorithmicx}
\usepackage{algcompatible}
\usepackage{algpseudocode}
\usepackage{color}
\usepackage{booktabs}

\usepackage[authoryear]{natbib}

\newcommand{\bea}{\begin{eqnarray*}}
\newcommand{\eea}{\end{eqnarray*}}
\newcommand{\be}{\begin{eqnarray}}
\newcommand{\ee}{\end{eqnarray}}

\def\dd{\mathrm{d}}

\def\Arg{\mathrm{Arg}}

\def\tr{\mathrm{trace}}

\def\mt{\theta}

\DeclareMathOperator{\CR}{\mathsf{CR}}

\def\Ex{\mathsf{E}}

\def\ma{\alpha}

\def\mt{\theta}

\def\ms{\sigma}

\def\Ex{\mathsf{E}}

\def\TT{^\top}
\def\e1{\mathsf{e}}

\def\Ab{\mathbf{A}}

\def\Kb{\mathbf{K}}

\def\Xb{\mathbf{X}}

\def\1b{\mathbf{1}}
\def\0b{\mathbf{0}}

\def\kb{\mathbf{k}}

\def\xb{\mathbf{x}}

\def\Yb{\mathbf{Y}}

\def\etab{\boldsymbol{\eta}}

\def\SG{\mathcal{G}}

\def\SI{\mathds{I}}

\def\SN{{\mathscr N}}
\def\SO{\mathcal{O}}

\def\SX{{\mathscr X}}

\newcommand{\thetav}{\boldsymbol{\theta}}

\newtheorem{thm}{Theorem}

\usepackage{collab}
\collabAuthor{wm}{orange!85}{Werner}
\collabAuthor{mh}{magenta!85}{Markus}
\collabAuthor{lp}{blue!85}{Luc}
\collabAuthor{hw}{green!85}{Henry}
\collabAuthor{ey}{red!85}{Elham}

\begin{document}

\title{Active Discrimination Learning for Gaussian Process Models}

\author{Elham Yousefi$^{1}$, Luc Pronzato$^{2}$, Markus Hainy$^{1}$, Werner G. Müller$^{1}$, Henry P. Wynn$^{3}$}
\date{$^{1}$JKU Linz, $^{2}$CNRS, $^{3}$London School of Economics}

\maketitle

\begin{abstract}
The paper covers the design and analysis of experiments to discriminate between two Gaussian process models, such as those widely used in computer experiments, kriging, sensor location and machine learning.
Two frameworks are considered. First, we study sequential constructions, where successive design (observation) points are selected, either as additional points to an existing design or from the beginning of observation. The selection relies on the maximisation of the difference between the symmetric Kullback Leibler divergences for the two models, which depends on the observations, or on the mean squared error of both models, which does not. Then, we consider static criteria, such as the familiar log-likelihood ratios and the Fr\'echet distance between the covariance functions of the two models. Other distance-based criteria, simpler to compute than previous ones, are also introduced, for which, considering the framework of approximate design, a necessary condition for the optimality of a design measure is provided. 
The paper includes a study of the mathematical links between different criteria and numerical illustrations are provided.
\end{abstract}


\section{Introduction}

The term `active learning' (cf.\ \cite{hino_active_2020} for a recent review) has replaced the traditional (sequential or adaptive) `design of experiments'  in the computer science literature, typically when the response is approximated by Gaussian process regression (GPR, cf. \cite{sauer_active_2022}). It refers to selecting the most suitable inputs to achieve the maximum of information from the outputs, usually with the aim of improving prediction accuracy. A good overview is given in Chapter 6 of \cite{gramacy_surrogates_2020}.

Frequently the aim of an experiment -- in the broad sense of any data acquisition exercise -- may rather be the discrimination between two or more potential explanatory models. When data can be sequentially collected during the experimental process, the literature goes back to the classic procedure of \cite{hunter_designs_1965} and has generated ongoing research (see  e.g.\ \cite{schwaab_sequential_2008}, \cite{olofsson_design_2018-1} and \cite{heirung_model_2019}).  When the design needs to be fixed before the experiment and thus no intermediate data will be available, the literature is less developed. While in the classical (non)linear regression case the criterion of T-optimality (cf.\ \cite{atkinson_design_1975}) and the numerous papers extending it was a major step, a similar breakthrough for Gaussian process regression is lacking.

With this paper we would like to investigate various sequential/adaptive and non-sequential design schemes for GPRs and their relative properties.
When the observations associated with the already collected points are available, one may base the criterion on the predictions and prediction errors (Section~\ref{S:3.1}). On the one hand, one natural choice will be to put the next design point where the symmetric Kullback-Leibler divergence between those two predictive (normal) distributions differs most. On the other hand, when the associated observations are not available, the incremental construction of the designs could be based on the mean squared error (MSE) for both models, assuming in turn that either of the two models is the true one (Section~\ref{sec:incr_uncond}).
{The static construction of a set of optimal designs of given size for nominal model parameters is the last mode we have considered (Section~\ref{S:distance-based-D}).
Our first choice is to use the difference between the expected values of the log likelihood ratios, assuming in turn that either of the two models is the true one.} This is actually a function of the symmetric Kullback-Leibler divergence, which also arises from Bayesian considerations. In a similar spirit, the Fr\'echet distance between two covariance matrices provides another natural criterion. Some further novel but simple approaches are considered in this paper as well. In particular we are interested whether complex likelihood-based criteria like the Kullback-Leibler-divergence can be effectively replaced by simpler ones based directly on the respective covariance kernels. The construction of optimal design measures for model discrimination (approximate design theory) is considered in Section~\ref{odm}.

Eventually, to compare the discriminatory power of the resulting designs from different criteria, one can compute the correct classification (hit) rates after selecting the model with the higher likelihood value.
A numerical illustration is provided in Section~\ref{S:Examples} for two Mat\'ern kernels with different smoothness.

\section{Notation}

One of the most popular design criteria for discriminating between rival models is T-optimality \citep{atkinson_design_1975}. This criterion is only applicable when the observations are independent and normally distributed with a constant variance. \cite{lopez2007optimal} generalised the normality assumption and developed an optimal discriminating design criterion to choose among non-normal models. The criterion is based on the log-likelihood ratio test under the assumption of independent observations. We denote by $\varphi_0(y,x,\theta_0)$ and $\varphi_1(y,x,\theta_1)$ the two rival probability density functions for one observation $y$ at point $x$. The following system of hypotheses might be considered:
\bea
H_0:\varphi(y,x)=\varphi_0(y,x,\theta_0)\\
H_1:\varphi(y,x)=\varphi_1(y,x,\theta_1)
\eea
where $\varphi_1(y,x,\theta_1)$ is assumed to be the true model. A common test statistic is the log-likelihood ratio given as
\bea
L=-\log\dfrac{\varphi_0(y,x,\theta_0)}{\varphi_1(y,x,\theta_1)}=\log\dfrac{\varphi_1(y,x,\theta_1)}{\varphi_0(y,x,\theta_0)},
\eea
where the null hypothesis is rejected when $\varphi_1(y,x,\theta_1)>\varphi_0(y,x,\theta_0)$ or equivalently
when $L>0$.
The power of the test refers to the expected value of the log-likelihood ratio criterion under the alternative hypothesis $H_1$. We have
\be \label{eq:PowerKLopt}
\Ex_{H_1}(L)=\Ex_{1}(L) &=&\int \varphi_1(y,x,\theta_1)\log\left\lbrace \dfrac{\varphi_1(y,x,\theta_1)}{\varphi_0(y,x,\theta_0)}\right\rbrace  \dd y\nonumber \\
&=&D_{KL}(\varphi_1\rVert \varphi_0),
\ee
where $D_{KL}(\varphi_1\rVert \varphi_0)$ is the Kullback–Leibler distance between the true and the alternative model \citep{kullback1951information}.

Interchanging the two models in the null and the alternative hypothesis, the power of the test would be
\be \label{eq:PowerKLopt2}
\Ex_{0}(-L) &=&D_{KL}(\varphi_0\rVert \varphi_1).
\ee

If it is not clear in advance which of the two models is the true model, one might consider to search for a design optimising a convex combination of \eqref{eq:PowerKLopt} and \eqref{eq:PowerKLopt2}, most commonly using weights $1/2$ for each model. This would be equivalent to maximising the symmetric Kullback-Leibler distance
\bea
D_{KL}(\varphi_0, \varphi_1) & = & \frac{1}{2} \left[ D_{KL}(\varphi_0\rVert \varphi_1) + D_{KL}(\varphi_1\rVert \varphi_0) \right].
\eea

In this paper we will consider random fields, i.e.\ we will allow for correlated observations. When the random field is Gaussian, we might still base the design strategy on the log-likelihood ratio criterion to choose among two rival models.

For a positive definite kernel $K(x,x')$ and an $n$-point design $\Xb_n=(x_1,\ldots,x_n)$, $\kb_n(x)$ is the $n$-dimensional vector $(K(x,x_1),\ldots,K(x,x_n))\TT$ and $\Kb_n$ is the $n\times n$ (kernel) matrix with elements $\{\Kb_n\}_{i,j}=K(x_i,x_j)$. Although $x$ is not bold, it may correspond to a point in a (compact) set $\SX \subset\mathds{R}^d$. Assume that $Y(x)$ corresponds to the realisation of a random field $Z_x$, indexed by $x$ in $\SX$, with zero mean $\Ex\{Z_x\}=0$ for all $x$ and covariance $\Ex\{Z_x Z_{x'}\}=K(x,x')$ for all $(x,x')\in\SX^2$. Our prediction of a future observation $Y(x)$ based on observations $\Yb_n=(Y(x_1),\ldots,Y(x_n))\TT$ corresponds to the best linear unbiased predictor (BLUP) $\widehat\eta_n(x)=\kb_n\TT(x)\Kb_n^{-1}\Yb_n$.
The associated prediction error is $e_n(x)=Y(x)-\widehat\eta_n(x)$ and we have
\bea
\Ex\{e_n^2(x)\} = \rho_n^2(x) = K(x,x)-\kb_n\TT(x)\Kb_n^{-1}\kb_n(x)\,.
\eea
The index $n$ will often be omitted when there is no ambiguity, and in that case $\kb_i(x)=\kb_{n,i}(x)$, $\Kb_i=\Kb_{n,i}$, $e_i(x)=e_{n,i}(x)$, $\rho_i^2(x)=\rho_{n,i}^2(x)$ will refer instead to model $i$, with $i\in\{0,1\}$.
We shall need to distinguish between the cases where the truth is model $0$ or model $1$,
and following \citet[p.~58]{stein_interpolation_1999} we denote by $\Ex_i$ the expectation computed with model $i$ assumed to be true.
We reserve the notation $\rho_i^2(x)$ to the case where the expectation is computed with the true model; i.e.,
\bea
\rho_i^2(x) = \Ex_i \{e_i^2(x)\} \,.
\eea
Hence we have $\rho_0^2(x) = \Ex_0\{e_0^2(x)\}= K_0(x,x)-\kb_0\TT(x)\Kb_0^{-1}\kb_0(x)$ and calculation gives
\be
\Ex_0\{e_1^2(x)\} &=& K_0(x,x)+\kb_1\TT(x)\Kb_1^{-1}\Kb_0\Kb_1^{-1}\kb_1(x)- 2\,\kb_1\TT(x)\Kb_1^{-1}\kb_0(x) \,, \label{E0e12} \\
\Ex_0\{[e_1(x)-e_0(x)]^2\} &=& \Ex_0\{e_1^2(x)\} - \Ex_0\{e_0^2(x)\} \,, \nonumber
\ee
with an obvious permutation of indices 0 and 1 when assuming the model 1 is true to compute $\Ex_1\{\cdot\}$.

If model 0 is correct, the prediction error is larger when we use model 1 for prediction than if we use the BLUP (i.e., model 0). \citet[p.~58]{stein_interpolation_1999} shows that the relation
\bea
\frac{\Ex_0\{e_1^2(x)\}}{\Ex_0\{e_0^2(x)\}} = 1+ \frac{\Ex_0\{[e_1(x)-e_0(x)]^2\}}{\Ex_0\{e_0^2(x)\}}
\eea
shown above is valid more generally for models with linear trends.
Also of interest is the assumed mean squared error (MSE) $\Ex_1\{e_1^2(x)\}$ when we use model 1 for assessing the prediction error (because we think it is correct) while the truth is model 0, and in particular the ratio
\bea
\frac{\Ex_1\{e_1^2(x)\}}{\Ex_0\{e_1^2(x)\}} = \frac{K_1(x,x)-\kb_1\TT(x)\Kb_1^{-1}\kb_1(x)}{\Ex_0\{e_1^2(x)\}} \,,
\eea
which may be larger or smaller than one.

Another important issue concerns the choice of covariance parameters in $K_0$ and $K_1$. Denote $K_i(x,x')=\ms_i^2\,C_{i,\mt_i}(x,x')$, $i=0,1$, $(x,x')\in\SX^2$, where the $\ms_i^2$ define the variance, the $\mt_i$ may correspond to correlation lengths in a translation invariant model and are thus scalar in the isotropic case, and $C(x,x')$ defines a correlation.

\section{Prediction-based discrimination}

For the incremental construction of a design for model discrimination, points are added conditionally on previous design points. We can distinguish the case where the observations associated with those previous points are available and can thus be used to construct a sequence of predictions (sequential, i.e., conditional, construction) from the unconditional case where observations are not used.

\subsection{Sequential (conditional) design}\label{S:3.1}

Consider stage $n$, where $n$ design points $\Xb_n$ and $n$ observations $\Yb_n$ are available. Assuming that the random field is Gaussian, when model $i$ is true we have $Y(x) \sim \SN(\widehat\eta_{n,i}(x),\rho_{n,i}^2(x))$. A rather natural choice is to choose the next design point $x_{n+1}$ where the symmetric Kullback-Leibler divergence between those two normal distributions differs most; that is,
\be \label{SeqCond}
x_{n+1} &\in& \Arg\max_{x\in\SX}  \frac{\rho_{n,0}^2(x)}{\rho_{n,1}^2(x)} + \frac{\rho_{n,1}^2(x)}{\rho_{n,0}^2(x)}
+ [\widehat\eta_{n,1}(x)-\widehat\eta_{n,0}(x)]^2\, \left[\frac{1}{\rho_{n,0}^2(x)}+ \frac{1}{\rho_{n,1}^2(x)} \right] \,.
\ee
Other variants could be considered as well, such as
\bea
x_{n+1} &\in& \Arg\max_{x\in\SX}\,  [\widehat\eta_{n,1}(x)-\widehat\eta_{n,0}(x)]^2 \,, \\
x_{n+1} &\in& \Arg\max_{x\in\SX}\,  \frac{[\widehat\eta_{n,1}(x)-\widehat\eta_{n,0}(x)]^2}{\rho_{n,0}^2(x)+\rho_{n,1}^2(x)} \,, \\
x_{n+1} &\in& \Arg\max_{x\in\SX}\,  [\widehat\eta_{n,1}(x)-\widehat\eta_{n,0}(x)]^2\, \left[\frac{1}{\rho_{n,0}^2(x)}+ \frac{1}{\rho_{n,1}^2(x)} \right] \,.
\eea
They will not be considered in the rest of the paper.

If necessary one can use plug-in estimates $\widehat\ms_{n,i}^2$ and $\widehat\mt_{n,i}$ of $\ms_i^2$ and $\mt_i$, for instance maximum likelihood (ML) or leave-one-out estimates based on $\Xb_n$ and $\Yb_n$, when we choose $x_{n+1}$. Note that the value of $\ms^2$ does not affect the BLUP $\widehat\eta_n(x)=\kb_n\TT\Kb_n^{-1}\Yb_n$.
In the paper we do not address the issues related to the estimation of $\ms^2$ or of the correlation length or smoothness parameters of the kernel; one may refer to \citet{karvonen_maximum_2020} 
and the recent papers \citet{Karvonen2022, KarvonenO2022}
for a detailed investigation. The connection between the notion of microergodicity, related to the consistency of the maximum-likelihood estimator, and discrimination through a KL divergence criterion is nevertheless considered in Example~1 below.

\subsection{Incremental (unconditional) design} \label{sec:incr_uncond}

Consider stage $n$, where $n$ design points $\Xb_n$ are available. We base the choice of the next point on the difference between the MSEs for both models, assuming that one or the other is true. For instance, assuming that model 0 is true, the difference between the MSEs is $\Ex_0\{e_1^2(x)\}-\Ex_0\{e_0^2(x)\}=\Ex_0\{[e_1(x)-e_0(x)]^2\}=\Ex_0 \{[\widehat\eta_{n,1}(x)-\widehat\eta_{n,0}(x)]^2\}$.

A first, un-normalised, version is thus
\be
\phi_A(x) &=& \Ex_0\{[e_1(x)-e_0(x)]^2\} + \Ex_1\{[e_1(x)-e_0(x)]^2\} \, \nonumber \\
&=& \Ex_0\{e_1^2(x)\}- \Ex_0\{e_0^2(x)\} + \Ex_1\{e_0^2(x)\}-\Ex_1\{e_1^2(x)\}\,. \label{un-normalised-incremental-1}
\ee
A normalisation seems in order here too, such as
\be\label{normalised-incremental-1}
\phi_B(x)=\frac{\Ex_0\{[e_1(x)-e_0(x)]^2\}}{\rho_{n,0}^2(x)} + \frac{\Ex_1\{[e_1(x)-e_0(x)]^2\}}{\rho_{n,1}^2(x)} = \frac{\Ex_0\{e_1^2(x)\}}{\Ex_0\{e_0^2(x)\}} + \frac{\Ex_1\{e_0^2(x)\}}{\Ex_1\{e_1^2(x)\}} - 2\,.
\ee

A third criterion is based on the variation of the symmetric Kullback-Leibler divergence \eqref{PhiKL} of Section~\ref{S:distance-based-D} when adding an $(n+1)$-th point $x$ to $\Xb_n$.
Direct calculation, using
\bea
\Kb_{n+1,i}=\left(
              \begin{array}{cc}
                \Kb_{n,i} & \kb_{n,i}(x) \\
                \kb_{n,i}\TT(x) & K_i(x,x) \\
              \end{array}
            \right)\,, \ i=0,1\,,
\eea
and the expression of the inverse of a block matrix, gives
\bea
\Phi_{KL\,[K0,K1]}(\Xb_n\cup\{x\}) = \Phi_{KL\,[K0,K1]}(\Xb_n) + \frac12\,\left[\frac{\Ex_1\{e_0^2(x)\}}{\Ex_0\{e_0^2(x)\}} + \frac{\Ex_0\{e_1^2(x)\}}{\Ex_1\{e_1^2(x)\}}\right] -1 \,.
\eea
We thus define
\be\label{phiKL}
\phi_{KL}(x)=\frac12\,\left[\frac{\Ex_1\{e_0^2(x)\}}{\Ex_0\{e_0^2(x)\}} + \frac{\Ex_0\{e_1^2(x)\}}{\Ex_1\{e_1^2(x)\}}\right] -1\,,
\ee
to be maximised with respect to $x\in\SX$.

Although the $\ms_i^2$ do not affect predictions, $\Ex_i\{e_j^2(x)\}$ is proportional to $\ms_i^2$. Unless specific information is available, it seems reasonable to assume that $\ms_0^2=\ms_1^2=1$. Other parameters $\mt_i$ should be chosen to make the two kernels the most similar, which seems easier to consider in the approach presented in Section~\ref{S:distance-based-D}, see \eqref{worst-case-Phi}. In the rest of this section we suppose that the parameters of both kernels are fixed.

The un-normalised version $\phi_A(x)$ given by \eqref{un-normalised-incremental-1} could be used to derive a one-step (non-incremental) criterion, in the same spirit as those of Section~\ref{S:distance-based-D}, through integration with respect to $x$ for a given measure $\mu$ on $\SX$. Indeed, we have
\bea
\Ex_0\{[e_1(x)-e_0(x)]^2\} = \kb_0\TT(x)\Kb_0^{-1}\kb_0(\xb)+ \kb_1\TT(x)\Kb_1^{-1}\Kb_0\Kb_1^{-1}\kb_1(x)-2\,\kb_1\TT(x)\Kb_1^{-1}\kb_0(x)\,,
\eea
so that
\bea
\int_\SX \Ex_0\{[e_1(x)-e_0(x)]^2\}\, \dd\mu(x) = \tr\left[\Kb_0^{-1}\Ab_0(\mu)+\Kb_1^{-1}\Kb_0\Kb_1^{-1}\Ab_1(\mu)-2\,\Kb_1^{-1}\Ab_{0,1}(\mu) \right] \,,
\eea
where $\Ab_i(\mu)=\int_\SX \kb_i(x)\kb_i\TT(x)\, \dd\mu(x)$, $i=0,1$, and $\Ab_{0,1}(\mu)=\int_\SX \kb_0(x)\kb_1\TT(x)\, \dd\mu(x)$. Similarly,
\bea
\int_\SX \Ex_1\{[e_1(x)-e_0(x)]^2\}\, \dd\mu(x) = \tr\left[\Kb_1^{-1}\Ab_1(\mu)+\Kb_0^{-1}\Kb_1\Kb_0^{-1}\Ab_0(\mu)-2\,\Kb_0^{-1}\Ab_{0,1}(\mu) \right] \,.
\eea
The matrices $\Ab_i(\mu)$ and $\Ab_{0,1}(\mu)$ can be calculated explicitly for some kernels and measures $\mu$. This happens in particular when $\SX=[0,1]^d$, the two kernels $K_i$ are separable, i.e., products of one-dimensional kernels on $[0,1]$, and $\mu$ is uniform on $\SX$.

\paragraph{Example~1: exponential covariance, no microergodic parameters.}

We consider Example~6 in \citet[p.~74]{stein_interpolation_1999} and take $K_i(x,x')=\e1^{-\ma_i|x-x'|}/\ma_i$, $i=0,1$. The example focuses on two difficulties: first, the two kernels only differ by their parameter values; second, the particular relation between the variance and correlation length makes the parameters $\ma_i$ not microergodic and they cannot be estimated consistently from observations on a bounded interval; see \citet[Chap.~6]{stein_interpolation_1999}. It is interesting to investigate the behaviour of the criteria \eqref{un-normalised-incremental-1}, \eqref{normalised-incremental-1} and \eqref{phiKL} in this particular situation.

We suppose that $n$ observations are made at $x_i=(i-1)/(n-1)$, $i=1,\ldots,n\geq 2$. We denote $\delta=\delta_n=1/[2(n-1)]$ the half-distance between two design points. The particular Markovian property of random processes with kernels $K_i$ simplifies the analysis. The prediction and MSE at a given $x\in(0,1)$ only depend on the position of $x$ relative to its two closest neighbouring design points; moreover, all other points have no influence. Therefore, due to the regular repartition of the $x_i$, we only need to consider the behaviour in one (any) interval $\SI_i=[a_i,b_i]=[x_i,x_{i+1}]$. 

We always have $\phi_A(x) \to 0$ as $x\to x_i\in\SX_n$. Numerical calculation shows that for $\delta_n$ small enough, $\phi_A(\cdot)$ has a unique maximum in $\SI_i$ at the centre $C_i=(x_i+x_{i+1})/2$. The next design point $x_{n+1}$ that maximises $\phi_A(\cdot)$ is then taken at $C_i$ for one of the $n-1$ intervals, and we get
\bea
\phi_A(C_i)= \frac14\, \frac{(\ma_1-\ma_0)^2(\ma_1+\ma_0)^3}{\ma_0\ma_1} \, \delta_n^4+ \SO(\delta_n^5) \,, \ n\to \infty \,.
\eea
Similar results apply to the case where the design $\Xb_n$ contains the endpoints 0 and 1 and its covering radius $\CR(\Xb_n)=\max_{x\in[0,1]}\min_{i=1,\ldots,n} |x-x_i|$ tends to zero, the points $x_i$ being not necessarily equally spaced: $C_i$ is then the centre of the largest interval $[x_i,x_{i+1}]$ and $\delta_n=\CR(\Xb_n)$.

When $\delta_n$ is large compared to the correlation lengths $1/\ma_0$ and $1/\ma_1$, there exist two maxima, symmetric with respect to $C_i$, that get closer to the extremities of $\SI_i$ as $\ma_1$ increases,
and $C_i$ corresponds to a local minimum of $\phi_A(\cdot)$. This happens for instance when $\ma_0\,\delta_n=1$ and $\ma_1\,\delta_n \gtrsim 2.600455$.

A similar behaviour is observed for $\phi_B(x)$ and $\phi_{KL}(x)$: for small enough $\delta_n$ they both have a unique maximum in $\SI_i$ at $C_i$, with now
\bea
\phi_B(C_i) &=& \frac14\, \frac{(\ma_1-\ma_0)^2(\ma_1+\ma_0)^3}{\ma_0\ma_1} \, \delta_n^3+ \SO(\delta_n^4) \,, \ n\to \infty \,,\\
\phi_{KL}(C_i) &=& \frac18\, \frac{(\ma_1-\ma_0)^2(\ma_1+\ma_0)^3}{\ma_0\ma_1} \, \delta_n^3+ \SO(\delta_n^4) \,, \ n\to \infty \,.
\eea
{Also, $\phi_B(x) \to 0$ and $\phi_{KL}(x) \to 0$ as $x\to x_i\in\Xb_n$.}
For large values of $\delta_n$ compared to the correlation lengths $1/\ma_0$ and $1/\ma_1$, there exist two maxima in $\SI_i$, symmetric with respect to $C_i$. When $\ma_0\,\delta_n=1$, this happens for instance when $\ma_1\,\delta_n \gtrsim 2.020178$ for $\phi_B(\cdot)$  and when $\ma_1\,\delta_n \gtrsim 7.251623$ for $\phi_{KL}(\cdot)$. However, in the second case the function is practically flat between the two maxima.

The left panel of Figure~\ref{F:Exl} presents $\phi_A(x)$, $\phi_B(x)$ and $\phi_{KL}(x)$ for $x\in[x_1,x_2]=[0,0.1]$ when $n=11$ ($\delta_n= 0.05$) and $\ma_0=1$, $\ma_1=10$. The right panel is for $\ma_0\,\delta_n=1$, $\ma_1\,\delta_n=10$.

\begin{figure}[ht!]
\begin{center}
\includegraphics[width=.49\linewidth]{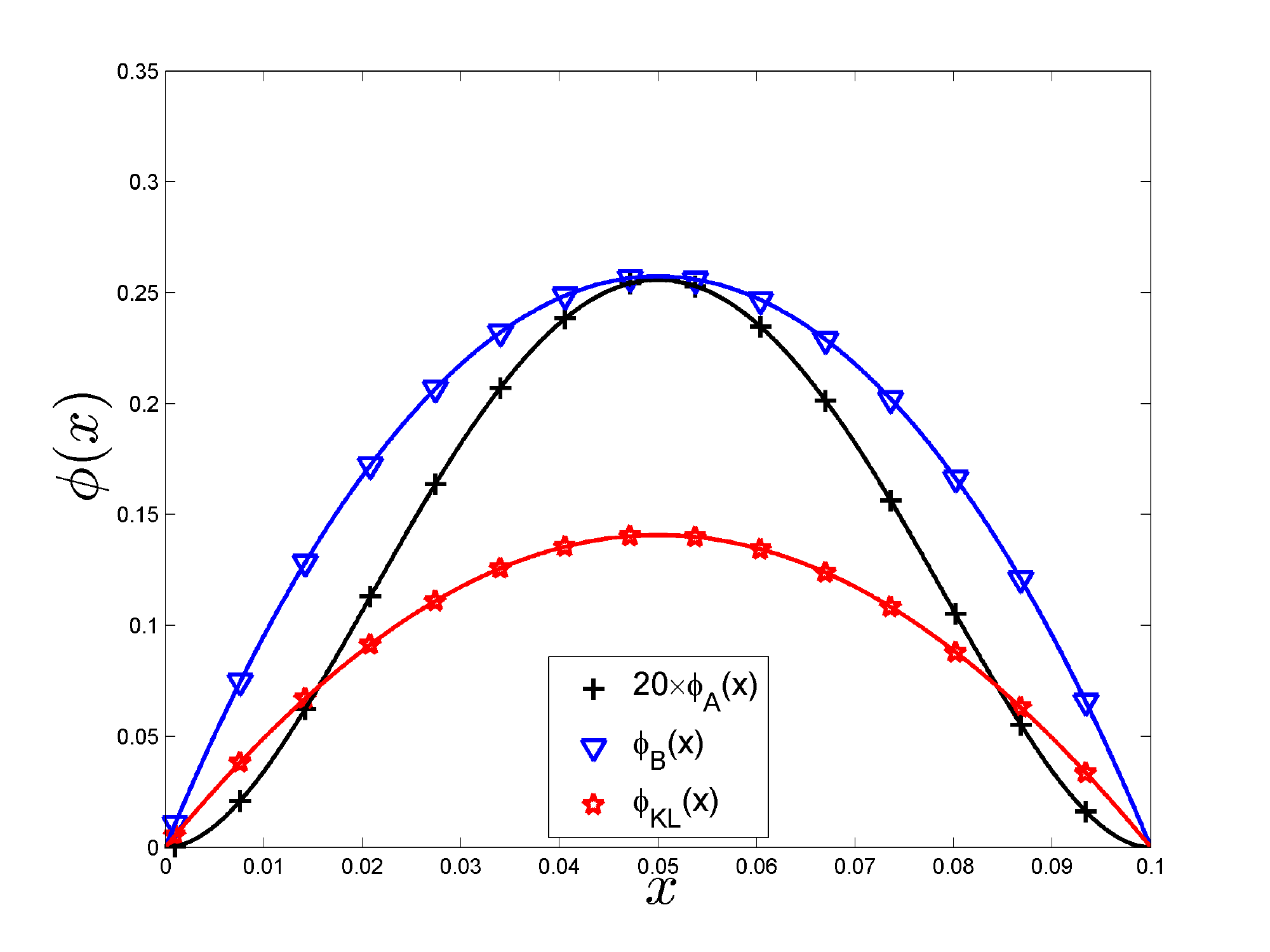} \includegraphics[width=.49\linewidth]{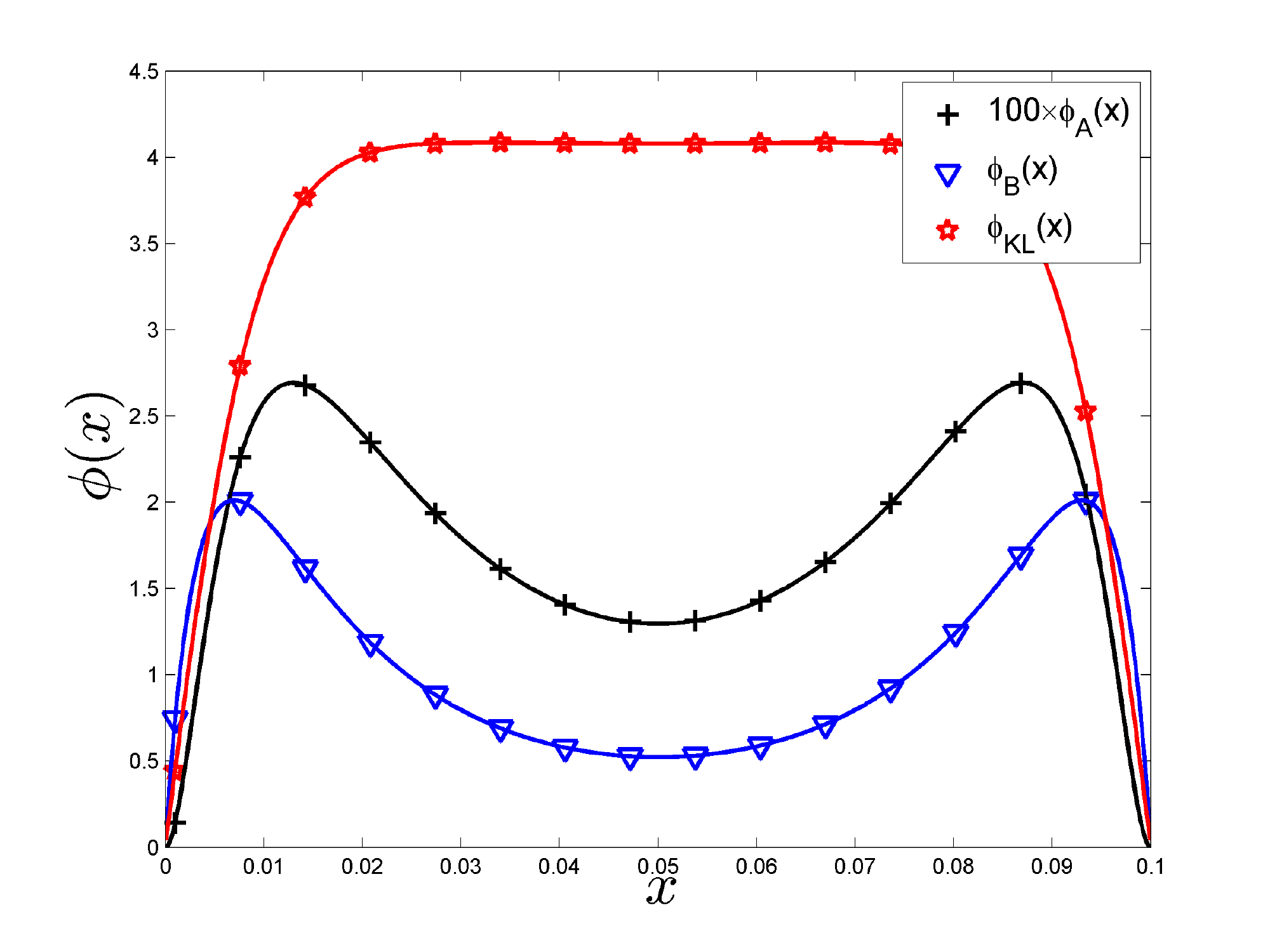}
\end{center}
\caption{\small $\phi_A(x)$, $\phi_B(x)$ and $\phi_{KL}(x)$, $x\in[x_1,x_2]$, for $n=11$ ($\delta_n= 0.05$) in Example~1. Left: $\ma_0=1$, $\ma_1=10$; Right: $\ma_0=20$, $\ma_1=200$.}\label{F:Exl}
\end{figure}

This behaviour of $\phi_{KL}(C_i)$ for small $\delta_n$ sheds light on the fact that $\ma$ is not estimable in this model. Indeed, consider a sequence of embedded $n_k$-point designs $\Xb_{n_k}$, initialised with the design $\Xb_n=\Xb_{n_0}$ considered above and with  $n_k=2^k\,(n_0-1)+1$, all these designs having the form $x_i=(i-1)/(n_k-1)$, $i=1,\ldots,n_k$. Then, $\CR(\Xb_{n_k})=\CR(\Xb_j)=\delta_j=1/[2(n_k-1)]$ for $j=n_k,\ldots,n_{k+1}-1=2\,n_k-2$. For $k$ large enough, the increase in Kullback-Leibler divergence \eqref{PhiKL} from $\Xb_{n_k}$ to $\Xb_{n_{k+1}}$ is thus bounded by $c/(n_k-1)^2$ for some $c>0$, so that the expected log-likelihood ratio $\Ex_0\{L_{n_k}\}-\Ex_1\{L_{n_k}\}$ remains bounded as $k\to\infty$.

More generally, denote by $0\leq x_1\leq x_2 \leq \cdots \leq x_n\leq 1$ the ordered points of an $n$-point design $\Xb_n$ in $[0,1]$, $n\geq 3$. Let $i^*\geq 3$ be such that $|x_{i^*-2}-x_{i^*}|=\min_{i=3,\ldots,n} |x_{i-2}-x_i|$. Then necessarily $|x_{i^*-2}-x_{i^*}|\leq 1/(\lceil n/2 \rceil-1)$. Indeed, consider the following iterative modification of $\Xb_n$ that cannot decrease $\min_{i=3,\ldots,n} |x_{i-2}-x_i|$: first, move $x_1$ to zero, then move $x_2$ to $x_1$; leave $x_3$ unchanged, but move $x_4$ to $x_3$, etc. For $n$ even, the design $\Xb_n'$ obtained is the duplication of an $(n/2)$-points design; for $n$ odd, only the right-most point $x_n$ remains single. In the fist case, the minimum distance between points of $\Xb_n'$ is at most $1/(n/2-1)$, in the second case it is at most $1/(\lceil n/2 \rceil-1)$.
We then define $\Xb_{n-1}=\Xb_n\setminus\{x_{i^*-1}\}$. For $n$ large enough, the increase in Kullback-Leibler divergence \eqref{PhiKL} from $\Xb_{n-1}$ to $\Xb_n$ is thus bounded by {$c/(\lceil n/2 \rceil-1)^3$ for some $c>0$} depending on $\ma_0$ and $\ma_1$. Starting from some design $\Xb_{n_0}$, we thus have, for $n_0$ large enough,
\bea
\Phi_{KL\,[K0,K1]}(\Xb_n)-\Phi_{KL\,[K0,K1]}(\Xb_{n_0}) \leq c \sum_{k=n_0+1}^n \frac{1}{(\lceil k/2 \rceil-1)^3} \,,
\eea
which implies $\lim_{n\to\infty} \Phi_{KL\,[K0,K1]}(\Xb_n) \leq B$ for some $B<\infty$. Assuming, without any loss of generality, that model 0 is correct, we have $0 \leq \Ex_0\{L_n\} \leq B$ (we get $0 \leq \Ex_1\{-L_n\} \leq B$ when we assume that model 1 is correct), implying in particular that $L_n$ does not tend to infinity a.s.\ and the ML estimator of $\ma$ is not strongly consistent.

\paragraph{Example~2: exponential covariance, microergodic parameters.}
{Consider now} two exponential covariance models with identical variances (which we take equal to one without any loss of generality): $K_i(x,x')=\e1^{-\ma_i|x-x'|}$, $i=0,1$.

Again, $\phi_A(x) \to 0$ as $x\to x_i\in\Xb_n$ and $\phi_A(\cdot)$ has a unique maximum at $C_i$ for small enough $\delta_n$, with now
\bea
\phi_A(C_i)= \frac12\, (\ma_1^2-\ma_0^2)^2 \, \delta_n^4+ \SO(\delta_n^5) \,, \ n\to \infty \,.
\eea
There are two maxima for $\phi_A(\cdot)$ in $\SI_i$, symmetric with respect to $C_i$ for large $\delta_n$: when $\ma_0\,\delta_n=1$, this happens for instance when $\ma_1\,\delta_n \gtrsim 2.558545$. Nothing is changed for $\phi_B(\cdot)$ compared to Example~1 as the variances cancel in the ratios that define $\phi_B(\cdot)$, see \eqref{E0e12} and \eqref{normalised-incremental-1}. The situation is quite different for $\phi_{KL}(\cdot)$, with
\bea
\phi_{KL}(C_i)= \frac12\, \frac{(\ma_1-\ma_0)^2}{\ma_0\ma_1} + \SO(\delta_n) \,, \ n\to \infty \,,
\eea
indicating that it is indeed possible to distinguish between the two models much more efficiently with this criterion than with the two others. Interestingly enough, the best choice for next design point is not at $C_i$ but always as close as possible to one of the endpoints $a_i$ or $b_i$, with however a criterion value similar to that in the centre $C_i$ when $\delta_n$ is small enough, as $\lim_{x\to x_i} \phi_{KL}(x) = (\ma_1-\ma_0)^2/(2\,\ma_0\ma_1)$.
Here, the same sequence of embedded designs as in Example~1 ensures that $\Ex_0\{L_{n_k}\}-\Ex_1\{L_{n_k}\} \to \infty$ as $k\to\infty$. Figure~\ref{F:Ex2} presents $\phi_A(x)$, $\phi_B(x)$ and $\phi_{KL}(x)$ in the same configuration as in Figure~\ref{F:Exl} but for the kernels $K_i(x,x')=\e1^{-\ma_i|x-x'|}$, $i=0,1$.

\begin{figure}[ht!]
\begin{center}
\includegraphics[width=.49\linewidth]{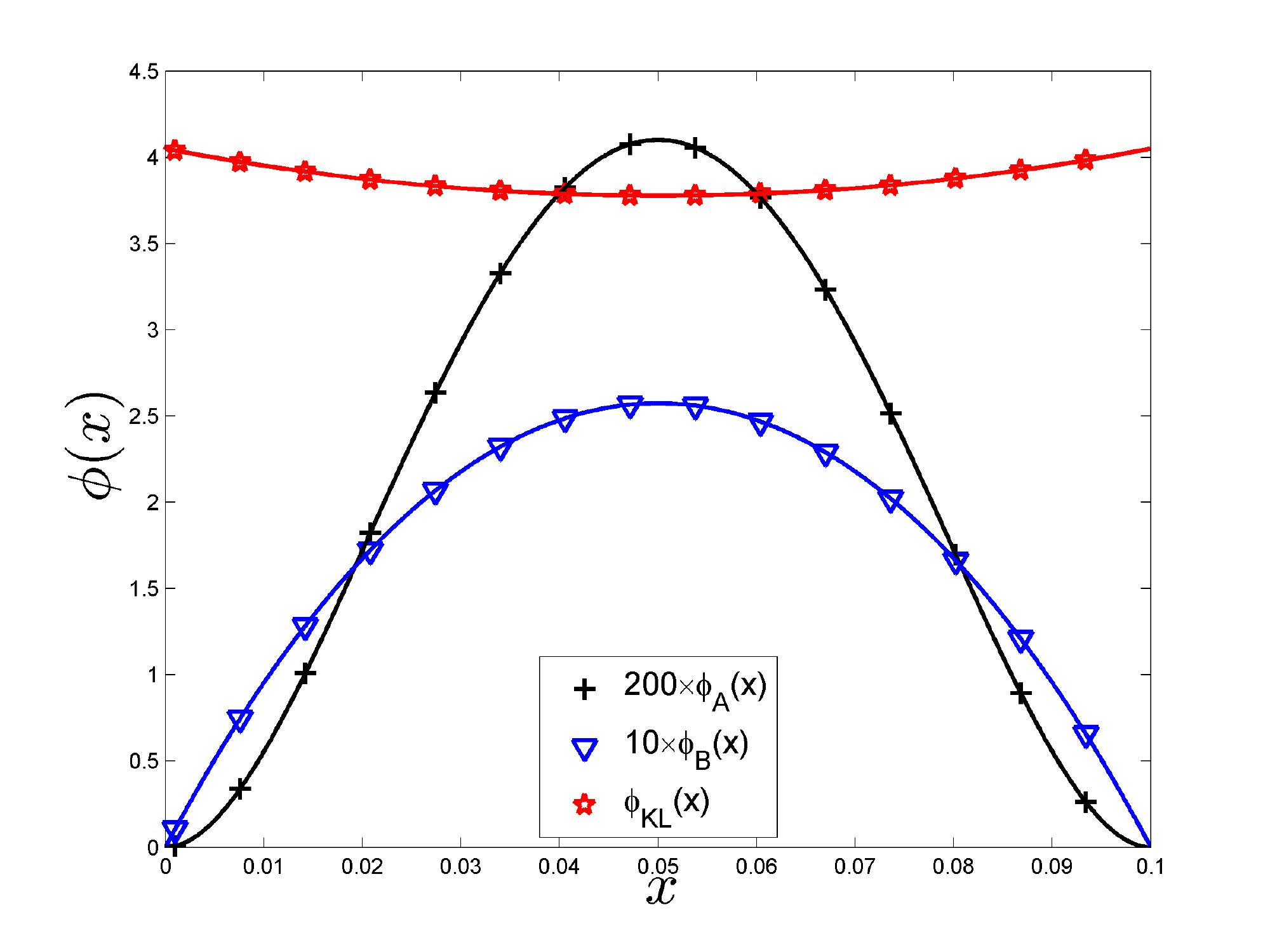} \includegraphics[width=.49\linewidth]{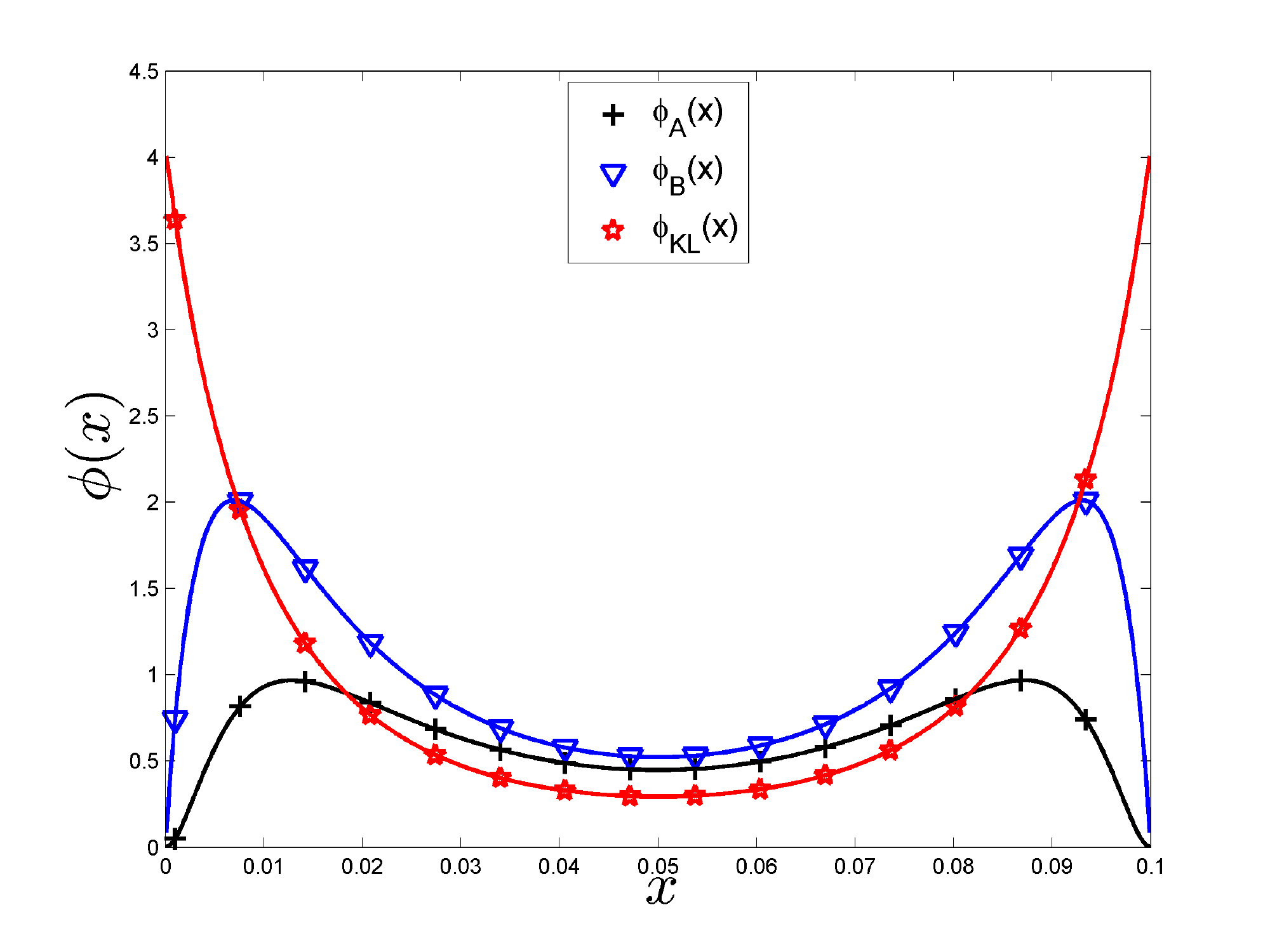}
\end{center}
\caption{\small $\phi_A(x)$, $\phi_B(x)$ and $\phi_{KL}(x)$, $x\in[x_1,x_2]$, for $n=11$ ($\delta_n= 0.05$) in Example~2. Left: $\ma_0=1$, $\ma_1=10$; Right: $\ma_0=20$, $\ma_1=200$.}\label{F:Ex2}
\end{figure}

\paragraph{Example~3: Mat\'ern kernels.}

Take $K_0$ and $K_1$ as the 3/2 and 5/2 Mat\'ern kernels, respectively:
\be
K_{0,\mt}(x,x') &=& (1+\sqrt{3}\mt\,|x-x'|)\, \exp(-\sqrt{3}\mt\,|x-x'|) \ \mbox{ (Mat\'ern 3/2)}\,, \label{Matern32}\\
K_{1,\mt}(x,x') &=& [1+\sqrt{5}\mt\, |x-x'|+ 5\mt^2\, |x-x'|^2/3]\,\exp(-\sqrt{5}\mt\, |x-x'|)  \ \mbox{ (Mat\'ern 5/2)}\,. \label{Matern52}
\ee
We take $\mt=\mt_0=1$ in $K_{0,\mt}$ and adjust $\mt=\mt_1$ in $K_{1,\mt}$ to minimise $\phi_{2\,[K_{0,\mt_0},K_{1,\mt_1}]}(\mu)$ defined by Eq.~\eqref{phi-psi} in Section~\ref{S:distance-based-D} with $\mu$ the uniform measure on $[0,1]$, which gives $\mt_1\simeq 1.1275$. The left panel of Figure~\ref{F:Ex3} shows
$K_{0,\mt_0=1}(x,0)$ and $K_{1,\mt}(x,0)$ for $\mt=1$ and $\mt=\mt_1$ when $x\in[0,1]$. The right panel presents $\phi_B(x)$ and $\phi_{KL}(x)$ for the same $n=11$-point equally spaced design $\Xb_n$ as in Example~1 and $x\in[0,1]$ for $K_{0,1}$ and $K_{1,1.1275}$ (the value of $\phi_A(x)$ does not exceed $0.65\,10^{-4}$ and is not shown). The behaviours of $\phi_B(x)$ and $\phi_{KL}(x)$ are now different in different intervals $[x_i,x_{i+1}]$ (they remain symmetric with respect to $1/2$, however), the maximum of $\phi_{KL}(x)$ is obtained at the central point $x_5$. The behaviour of $\phi_{KL}(\cdot)$ could be related to the fact that discriminating between $K_0$ {and $K_1$} amounts to estimating the smoothness of the realisation, which requires that some design points are close to each other.

\begin{figure}[ht!]
\begin{center}
\includegraphics[width=.49\linewidth]{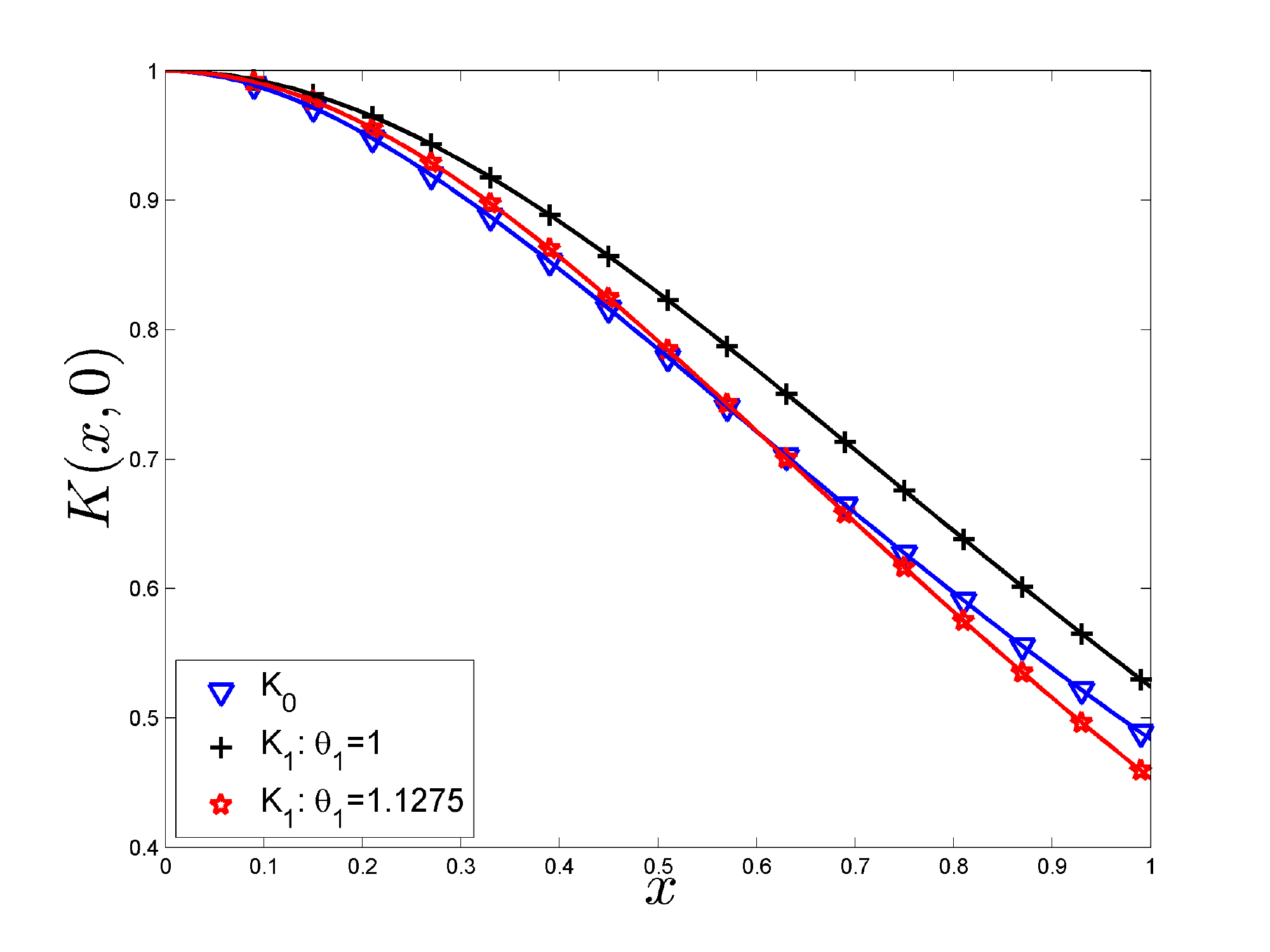} \includegraphics[width=.49\linewidth]{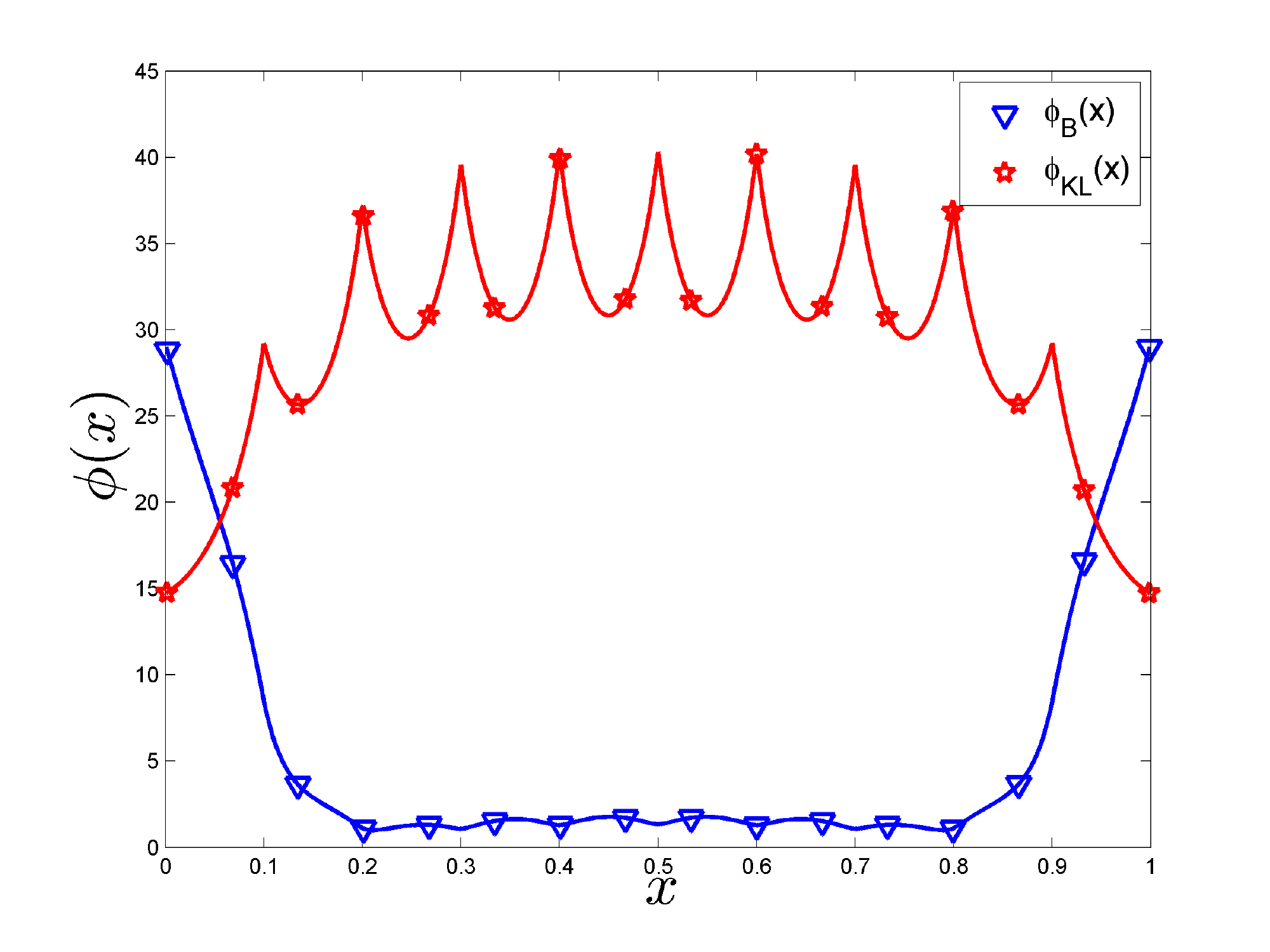}
\end{center}
\caption{\small Left: $K_{0,1}(x,0)$, $K_{1,1}(x,0)$ and $K_{1,1.1275}(x,0)$, $x\in[0,1]$. Right: $\phi_B(x)$ and $\phi_{KL}(x)$ for $x\in[0,1]$ and the same $11$-point equally spaced design {$\Xb_n=\{0,1/10,2/10,\ldots,1\}$} as in Example~1, with $K_{0,1}$ and $K_{1,1.1275}$.}\label{F:Ex3}
\end{figure}

\section{Distance-based discrimination}\label{S:distance-based-D}

We will now consider criteria which are directly based on the discrepancies of the covariance kernels. Ideally those should be simpler to compute and still exhibit reasonable efficiencies and some similar properties. The starting point is again the use of the log-likelihood ratio criterion to choose among the two models. Assuming that the random field is Gaussian, the probability densities of observations $\Yb_n$ for the two models are
\bea
\varphi_{n,i}(\Yb_n) = \frac{1}{(2\pi)^{n/2}\,\det^{1/2} \Kb_{n,i}} \, \exp\left[ - \frac12\, \Yb_n\TT\Kb_{n,i}^{-1} \Yb_n \right]\,, \ i=0,1\,.
\eea
The expected value of the log-likelihood ratio $L_n=\log \varphi(\Yb_n|0)- \log\varphi(\Yb_n|1)$  
under model 0 is
\bea
\Ex_0\{L_n\} = \frac12 \log\det(\Kb_{n,1}\Kb_{n,0}^{-1}) - \frac{n}{2} +\frac12\, \tr(\Kb_{n,0}\Kb_{n,1}^{-1})
\eea
and similarly
\bea
\Ex_1\{L_n\} = \frac12 \log\det(\Kb_{n,1}\Kb_{n,0}^{-1}) + \frac{n}{2} - \frac12\, \tr(\Kb_{n,1}\Kb_{n,0}^{-1}) \,.
\eea
A good discriminating design should make the difference $\Ex_0\{L_n\}-\Ex_1\{L_n\}$ as large as possible; that is, we should choose $\Xb_n$ that maximises
\be
\Phi_{KL\,[K_0,K_1]}(\Xb_n)= \Ex_0\{L_n\}-\Ex_1\{L_n\} &=& \frac12 \left[\tr(\Kb_{n,0}\Kb_{n,1}^{-1}) + \tr(\Kb_{n,1}\Kb_{n,0}^{-1}) \right] - n \, \nonumber \\
&=&  2\, D_{KL}(\varphi_{n,0},\varphi_{n,1})\,, \label{PhiKL}
\ee
i.e.\ twice the symmetric Kullback-Leibler divergence between the normal distributions with densities $\varphi_{n,0}$ and $\varphi_{n,1}$.

We may enforce the normalisation $\ms_0^2=\ms_1^2=1$ and choose the $\mt_i$ to make the two kernels most similar in the sense of the criterion $\Phi(\cdot,\cdot)$ considered; that is, maximise
\be \label{worst-case-Phi}
\min_{\mt_0\in\Theta_0,\,\mt_1\in\Theta_1}\Phi_{KL\,[K_0,K_1]}(\Xb_n) \,.
\ee
The choice of $\Theta_0$ and $\Theta_1$ is important; in particular, unconstrained minimisation over the $\mt_i$ could make both kernels completely flat or on the opposite close to Dirac distributions. It may thus be preferable to fix $\mt_0$ and minimise over $\mt_1$ without constraints. Also, the
Kullback-Leibler distance is sensitive to kernel matrices being near singularity, which might happen if design points are very close to each other. \cite{pronzato_bregman_2019} suggest a family of criteria based on matrix distances derived from Bregman divergences between functions of covariance matrices from Kiefer's $\varphi_p$-class of functions \citep{kiefer_equivalence_1974}. If $p \in (0,1)$, these criteria are rather insensitive to eigenvalues close or equal to zero. Alternatively, they suggest criteria computed as Bregman divergences between squared volumes of random $k$-dimensional simplices for $k \in \{2,\ldots,d-1\}$, which have similar properties.

The index $n$ is omitted in the following and we consider fixed parameters for both kernels. The Fr\'echet-distance criterion
\be\label{Phi-Frechet}
\Phi_{F\,[K_0,K_1]}(\Xb_n) = \tr\left[ \Kb_0+\Kb_1-2\,(\Kb_0\Kb_1)^{1/2}\right],
\ee
related to the Kantorovich (Wasserstein) distance, seems of particular interest due to the absence of matrix inversion. The expression is puzzling since the two matrices do not necessarily commute, but the paper \cite{dowson_frechet_1982} is illuminating.

Other matrix ``entry-wise" distances will be considered, in particular the one based on the (squared) Frobenius norm,
\bea
\Phi_{2\,[K_0,K_1]}(\Xb_n) = \tr\left( \Kb_0^2+\Kb_1^2-2\,\Kb_0\Kb_1\right) = \tr\left[(\Kb_0-\Kb_1)^2\right],
\eea
which corresponds to the substitution of $\Kb_i^2$ for $\Kb_i$ in \eqref{Phi-Frechet} for $i=0,1$.
Denote more generally
\bea
\Phi_{p\,[K_0,K_1]}(\Xb_n) = \|\Kb_1-\Kb_0\|_p^p = \sum_{i,j=1}^n |\{\Kb_1-\Kb_0\}_{i,j}|^p = \1b_n\TT |\Kb_1-\Kb_0|^{\odot p} \1b_n \,, \ p>0 \,,
\eea
where $\1b_n$ is the $n$-dimensional vector with all components equal to 1, the absolute value is applied entry-wise and $^{\odot p}$ denotes power $p$ applied entry-wise.

Figure~\ref{F:Ex4} shows the values of the criteria $\Phi_{i\,[K_{0,1},K_{1,\mt}]}$, $i=1,2$, $\Phi_{F\,[K_{0,1},K_{1,\mt}]}$ and $\Phi_{KL\,[K_{0,1},K_{1,\mt}]}$ as functions of $\mt$ for the two kernels $K_{0,\mt}$ and $K_{1,\mt}$ given by \eqref{Matern32} and \eqref{Matern52} and the same {regular design} as in Example~1: $x_i=(i-1)/(n-1)$, $i=1,\ldots,11$. The criteria are re-scaled so that their maximum equals one on the interval considered for $\mt$.  Note the similarity between $\Phi_{2\,[K_{0,1},K_{1,\mt}]}(\Xb_n)$ and $\Phi_{F\,[K_{0,1},K_{1,\mt}]}(\Xb_n)$ and the closeness between the distance-minimising $\mt$ for $\Phi_1$, $\Phi_2$ and $\Phi_F$. Also note the good agreement  with the value $\mt_1\simeq 1.1275$ that minimises $\phi_{2\,[K_{0,1},K_{1,\mt_1}]}(\mu)$ from Eq.~\eqref{phi-psi}, see Example~3. The optimal $\mt$ for $\Phi_{KL\,[K_{0,1},K_{1,\mt}]}(\Xb_n)$ is much different, however, showing that the criteria do not necessarily agree between them.

\begin{figure}[ht!]
\begin{center}
\includegraphics[width=.49\linewidth]{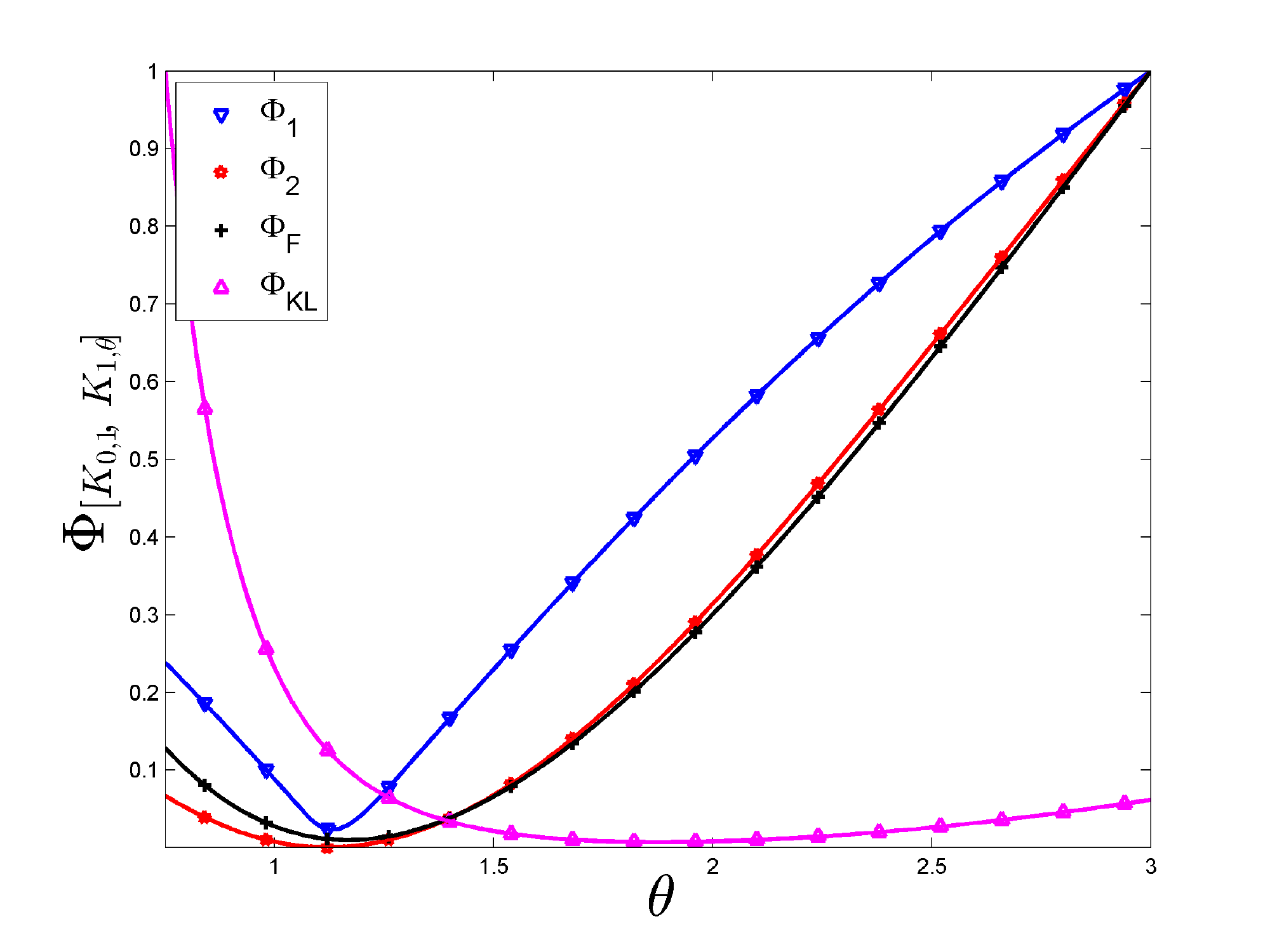} \end{center}
\caption{\small $\Phi_{i\,[K_{0,1},K_{1,\mt}]}(\Xb_n)$, $i=1,2$, $\Phi_{F\,[K_{0,1},K_{1,\mt}]}(\Xb_n)$ and $\Phi_{KL\,[K_{0,1},K_{1,\mt}]}(\Xb_n)$ as functions of $\mt\in[0.75,3]$ for the same $11$-point equally spaced design $\Xb_n$ as in Example~1 and $K_{0,\mt}$, $K_{1,\mt}$ given by \eqref{Matern32} and \eqref{Matern52}, respectively.}\label{F:Ex4}
\end{figure}

An interesting feature of the family of criteria $\Phi_{p\,[K_0,K_1]}(\cdot)$, $p>0$, is that they extend straightforwardly to a design measure version. Indeed, defining $\xi_n$ as the empirical measure on the points in $\Xb_n$, $\xi_n=(1/n)\,\sum_{i=1}^n \delta_{x_i}$, we can write
\bea
\Phi_{p\,[K_0,K_1]}(\Xb_n) = n^2\, \phi_{p\,[K_0,K_1]}(\xi_n) \,,
\eea
where we define, for any design (probability) measure on $\SX$,
{
\be
\phi_p(\xi) = \phi_{p\,[K_0,K_1]}(\xi) &=& \int_{\SX^2} |K_1(x,x')-K_0(x,x')|^p \, \dd\xi(x)\dd\xi(x') \,.\label{phi-psi}
\ee
}

Denote by $F_{p\,[K_0,K_1]}(\xi;\nu)$ the directional derivative of $\phi_{p\,[K_0,K_1]}(\cdot)$ at $\xi$ in the direction $\nu$,
\bea
F_{p\,[K_0,K_1]}(\xi;\nu) = \lim_{\ma\ra 0^+} \frac{\phi_{p\,[K_0,K_1]}[(1-\ma)\xi+\ma\nu]-\phi_{p\,[K_0,K_1]}(\xi)}{\ma} \,.
\eea
Direct calculation gives
\bea
F_{p\,[K_0,K_1]}(\xi;\nu) = 2 \left[ \int_{\SX^2} |K_1(x,x')-K_0(x,x')|^p\, \dd\nu(x)\dd\xi(x') - \phi_{p\,[K_0,K_1]}(\xi) \right] \,,
\eea
and thus in particular
\bea
F_{p\,[K_0,K_1]}(\xi;\delta_x) = 2 \left[ \int_\SX |K_1(x,x')-K_0(x,x')|^p\, \dd\xi(x') - \phi_{p\,[K_0,K_1]}(\xi) \right] \,.
\eea
One can easily check that the criterion is neither concave nor convex in general (as the matrix $|\Kb_1-\Kb_0|^{\odot p}$ can have both positive and negative eigenvalues), but we nevertheless have a necessary condition for optimality.

\begin{thm}\label{Th:NC} If the probability measure $\xi^*$ on $\SX$ maximises $\phi_{p\,[K_0,K_1]}(\xi)$, then
\bea
\forall x\in\SX, \  \int_\SX |K_1(x,x')-K_0(x,x')|^p\, \dd\xi^*(x') \leq \phi_{p\,[K_0,K_1]}(\xi^*) \,.
\eea
Moreover, $\int_\SX |K_1(x,x')-K_0(x,x')|^p\, \dd\xi^*(x')=\phi_{p\,[K_0,K_1]}(\xi^*)$ for $\xi^*$-almost every $x\in\SX$. \label{Theorem:necessaryCon_Phip}
\end{thm}

This suggests the following simple incremental construction: at iteration $n$, with $\Xb_n$ the current design and $\xi_n$ the associated empirical measure, choose $x_{n+1}\in\Arg\max_{x\in\SX} F_{p\,[K_0,K_1]}(\xi_n;\delta_x)=\1b_n\TT|\kb_{n,0}(x)-\kb_{n,1}(x)|^{\odot p}$. It will be used in the numerical example of Section~\ref{S:6.1.1}.

\section{Optimal design measures}\label{odm}

In this section we explain why the determination of optimal design measures maximising $\phi_p(\xi)$ is generally difficult, even when limiting ourselves to the satisfaction of the necessary condition in Theorem~\ref{Th:NC}. At the same time, we can characterise measures that are approximately optimal for large $p$. 

We assume that the two kernels { are isotropic, i.e., such that $K_i(x,x')=\Psi_i(\|x-x'\|)$, $i=0,1$, and that the functions $\Psi_i$ are differentiable except possibly at 0 where they only admit a right derivative. We define {$\psi(t)=|\Psi_1(t)-\Psi_0(t)|$,} $t\in\mathds{R}^+$, and assume that the kernels} have been normalised so that $K_0(x,x)=K_1(x,x)$; that is, $\psi(0)=0$. Also, we only consider the case where the function $\psi(\cdot)$ has a unique global maximum on $\mathds{R}^+$. This assumption is not very restrictive. Consider again the two Mat\'ern kernels (\ref{Matern32}) and (\ref{Matern52}).
Figure~\ref{F:K32-K52} shows the evolution of $\psi^2(t)$ for $K_0=K_{0,1}$ and $K_1=K_{1,\mt_1}$ with two different values of $\mt_1$: $\mt_1=1$ and $\mt_1\simeq 1.1275$; the latter minimises $\phi_{2\,[K_{0,1},K_{1,\mt}]}(\mu)$ for $\mu$ being the uniform measure on $[0,1]$.

\begin{figure}[ht!]
\begin{center}
\includegraphics[width=.49\linewidth]{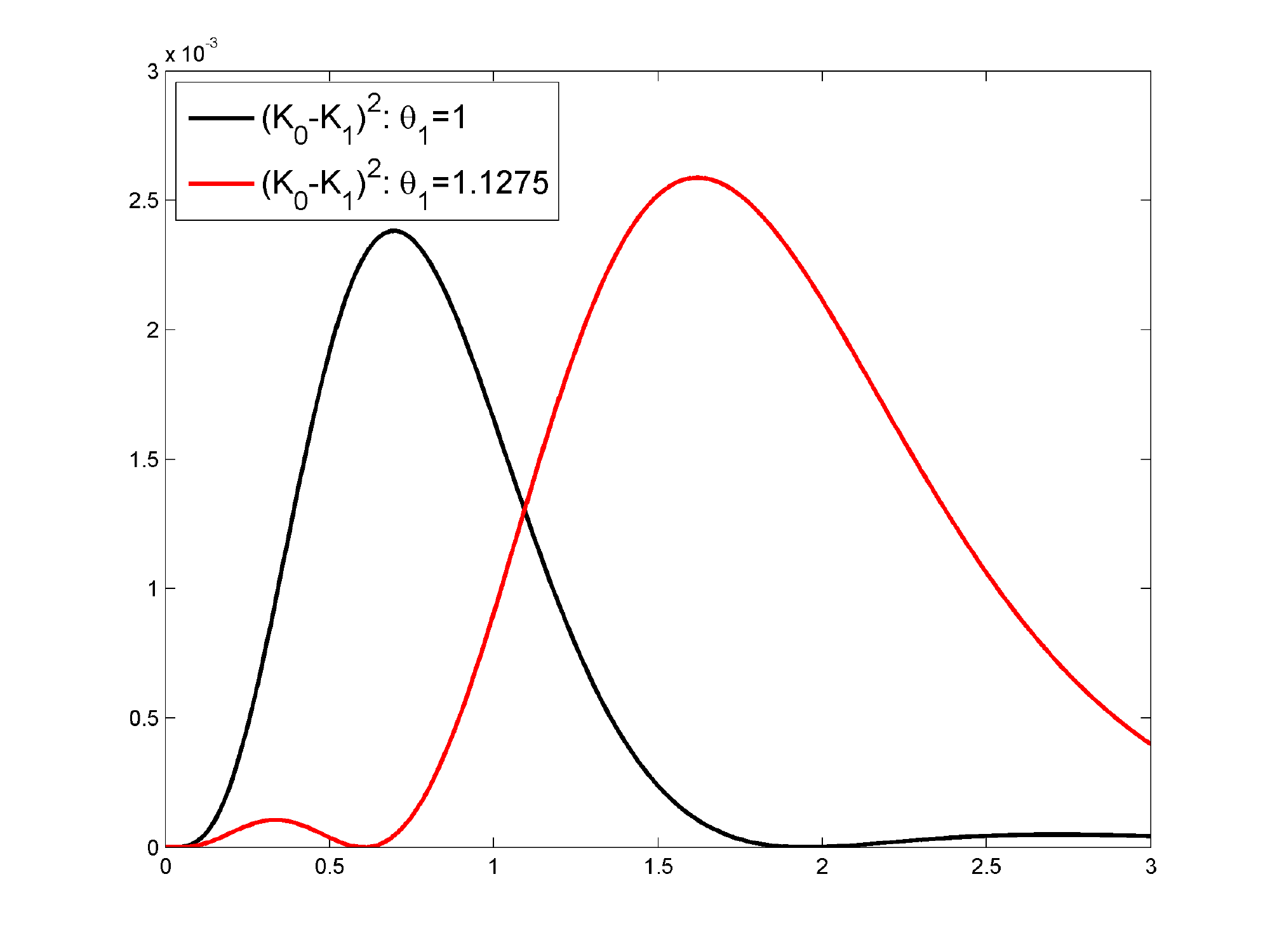}
\end{center}
\caption{\small $\psi^2(t)$ for $K_0=K_{0,1}$ and $K_1=K_{1,\mt_1}$ with two different values of $\mt_1$.}\label{F:K32-K52}
\end{figure}

In the following, we shall consider normalised functions $\psi(\cdot)$, such that $\max_{t\in\mathds{R}^+} \psi(t)=1$.
We denote by $\Delta$ the (unique) value such that $\psi(\Delta)=1$. On Figure~\ref{F:K32-K52}, $\Delta\simeq 0.7$ when $K_1=K_{1,1}$.

\subsection{A simplified problem with an explicit optimal solution}\label{S:simplified}

Consider the extreme case where $\psi=\psi_*$ defined by
\be
\psi_*(t) = \left\{ \begin{array}{ll}
1 & \mbox{if } t=\Delta, \\
0 & \mbox{otherwise.}
\end{array}\right. \label{psi_*}
\ee
Note that $\psi_*^p(t)=\psi_*(t)$ for any $p>0$; we can thus restrict our attention to $p=1$ for the maximisation of $\phi_p(\xi)$ {defined by \eqref{phi-psi}; that is, we consider
\bea
\phi_1(\xi) = \int_{\SX^2} \psi_*(\|x-x'\|) \, \dd\xi(x)\dd\xi(x') \,.
\eea
}

\begin{thm}\label{Th:optimal-simplified}
When $\psi=\psi_*$ and $\SX\subset\mathds{R}^d$ is large enough to contain a regular $d$ simplex with edge length $\Delta$, any measure $\xi^*$ allocating weight $1/(d+1)$ at each vertex of such a simplex maximises $\phi_1(\xi)$, and $\phi_1(\xi^*)=d/(d+1)$.
\end{thm}

\begin{proof}
Since $\phi_1(\xi)=0$ when $\xi$ is continuous with respect to the Lebesgue measure on $\SX$, we can restrict our attention to measures without any continuous component. Assume that $\xi=\sum_{i=1}^n w_i \delta_{x_i}$, with $w_i\geq 0$ for all $i$ and $\sum_{i=1}^n w_i=1$, $n\in\mathds{N}$. 
Consider the graph $\SG(\xi)$ having the $x_i$ as vertices, with an edge $(i,j)$ connecting $x_i$ and $x_j$ if and only if $\|x_i-x_j\| = \Delta$. We have 
\bea
\phi_1(\xi) = \sum_{(i,j)\in\SG(\xi)} w_i w_j,
\eea
and Theorem~1 of \cite{motzkin_maxima_1965} implies that $\phi_1(\xi)$ is maximum when $\xi$ is uniform on the maximal complete subgraph of $\SG(\xi)$. The maximal achievable order is $d+1$, obtained when the $x_i$ are the vertices of a regular simplex in $\SX$ with edge length $\Delta$. \cite{motzkin_maxima_1965} also indicate in their Theorem~1 that $\phi_1(\xi^*)=1-1/(d+1)$. This is easily recovered knowing that $\SG(\xi^*)$ is fully connected with order $d+1$. Indeed, we then have
\bea
\phi_1(\xi)=\sum_{i=1}^{d+1} w_i \sum_{\stackrel{j=1}{j\neq i}}^{d+1} w_j = \sum_{\stackrel{i,j=1}{j\neq i}}^{d+1} w_iw_j = 1-\sum_{i=1}^{d+1} w_i^2,
\eea
which is maximum when all $w_i$ equal $1/(d+1)$.
\end{proof}

\subsection{Optimal designs for  $\psi(t)=|\Psi_1(t)-\Psi_0(t)|$}
The optimal designs of Theorem~\ref{Th:optimal-simplified} are natural candidates for being optimal when we return to the case of interest $\psi(t)=|\Psi_1(t)-\Psi_0(t)|$. In the light of Theorem~\ref{Th:NC}, for a given probability measure $\xi$ on $\SX$, we consider the function
\bea
\delta_\xi(x)=\int_\SX \psi^p(\|x-x'\|)\, \dd\xi(x') - \phi_p(\xi),
\eea
which must satisfy $\delta_\xi(x)\leq 0$ for all $x\in\SX$ when $\xi$ is optimal. For an optimal measure $\xi^*$ as in Theorem~\ref{Th:optimal-simplified}, with support $x_1,\ldots,x_{d+1}$ forming a regular $d$-simplex, we have
\bea
\delta_{\xi^*}(x)= \frac{1}{d+1} \, \left[\sum_{i=1}^{d+1} \psi^p(\|x-x_i\|) - d \right].
\eea
One can readily check that $\delta_{\xi^*}(x_i)=0$ for all $i$ (as $\psi(\|x_i-x_j\|)=\psi(\Delta)=1$ for $i\neq j$ and $\psi(0)=0$). Moreover, {since} $\psi(\cdot)$ is differentiable everywhere except possibly at zero, when {$p>1$
the gradient of $\delta_{\xi^*}(x)$ equals zero at each $x_i$.} However, these $d+1$ stationary points may sometimes correspond to local minima --- a situation when of course $\xi^*$ is not optimal. The left panel of Figure~\ref{F:delta_x} shows an illustration ($d=2$) for $p=1.5$, $K_0(x,x')=\exp(-\|x-x'\|)$ and $K_1$ being the Mat\'ern 5/2 kernel $K_{1,1}$. The measure $\xi^*$ is supported at the vertices of the equilateral triangle $(0,0),(\Delta,0),(\Delta/2,\sqrt(3)\Delta/2)$ (indicated in blue on the figure), with $\Delta\simeq  0.53$ (the value where $\psi(\cdot)$ is maximum). Here the $x_i$ correspond to local minima of $\delta_{\xi^*}(x)$, {$\psi(\cdot)$ is not differentiable at zero but $p>1$ so that $\delta_{\xi^*}(\cdot)$ is differentiable.}

When $p\to\infty$, $\psi^p(\cdot)$ approaches the (discontinuous) function $\psi_*(\cdot)$, suggesting that $\xi^*$ may become close to being optimal for $\phi_p$ when $p$ is large enough. However, when $\SX$ is large, $\xi^*$ is never truly optimal, no matter how large $p$ is. Indeed, suppose that $\SX$ contains a point $x_*$ corresponding to the symmetric of a vertex $x_k$ of the simplex defining the support of $\xi^*$ with respect to the opposite face of that simplex.
Direct calculation gives
\bea
L=\|x_k-x_*\|=2\,\Delta\,\left(\frac{d+1}{2\,d}\right)^{1/2}.
\eea
The right panel of Figure~\ref{F:delta_x} shows an illustration for $K_0$ and $K_1$ being the Mat\'ern 3/2 and Mat\'ern 5/2 kernels $K_{0,1}$ and $K_{1,1}$, respectively. The measure $\xi^*$ is supported at the vertices of the equilateral triangle with vertices $(0,0),(\Delta,0),(\Delta/2,\sqrt(3)\Delta/2)$ with now $\Delta\simeq 0.7$. {At the point $x_*$, symmetric to $x_k$, indicated in red on the figure,} we have
\be
\delta_{\xi^*}(x_*) &=& \frac{1}{d+1} \, \left[\sum_{\stackrel{i=1}{i\neq k}}^{d+1} \psi^p(\|x_*-x_i\|) + \psi^p(\|x_*-x_k\|) - d \right] \nonumber \\
&=& \frac{1}{d+1}\, \psi^p(L) >0 \,, \label{xi^*-not-opt}
\ee
where the second equality follows from $\|x_*-x_i\|=\Delta$ for all $i\neq k$, implying that $\xi^*$ is not optimal.
Another, more direct, proof of the non-optimality of $\xi^*$ is to consider the measure $\widehat\xi$ that sets weights $1/(d+1)$ at all $x_i\neq x_k$ and weights $1/[2(d+1)]$ at $x_k$ and its symmetric $x_*$.
Direct calculation gives
\bea
\phi_p(\widehat\xi) = \frac{d}{d+1}\,\left(1-\frac{1}{d+1}\right) + \frac{2}{2\,(d+1)}\left[\frac{d}{d+1}+\frac{1}{2\,(d+1)} \,\psi^p(L)\right].
\eea
The first term on the right-hand side comes from the $d$ vertices $x_i$, $i\neq k$, each one having weight $1/(d+1)$ and being at distance $\Delta$ of all other vertices, those having total weight $1-1/(d+1)$. 
The second term comes from the two symmetric points $x_k$ and $x_*$, each one with weight $1/[2(d+1)]$. Each of these two points is at distance $\Delta$ from $d$ vertices with weights $1/(d+1)$ and at distance $L$ of the other opposite point with weight $1/[2(d+1)]$. We get after simplification
\bea
\phi_p(\widehat\xi) = \frac{d}{d+1} + \frac{\psi^p(L)}{2\,(d+1)^2} > \phi_p(\xi^*) = \frac{d}{d+1} ,
\eea
showing that $\xi^*$ is not optimal. Note that, for symmetry reasons, {the design} $\widehat\xi$ is not optimal for large enough $\SX$. The determination of a truly optimal design seems very difficult. In the simplified problem of Section~\ref{S:simplified}, where the criterion is based on the function $\psi_*$ defined by \eqref{psi_*}, the measures $\xi^*$ and $\widehat\xi$ supported on $d+1$ and $d+2$ points, respectively, have the same criterion value
$\phi_p(\xi^*)=\phi_p(\widehat\xi)=d/(d+1)$ for all $p>0$.

Although $\xi^*$ is not optimal, since $\psi(\|x_*-x_k\|)<1$ (as $\psi(t)$ takes its maximum value 1 for $t=\Delta$), \eqref{xi^*-not-opt} suggests that $\xi^*$ may be only marginally suboptimal when $p$ is large enough. Moreover, as the right panel of Figure~\ref{F:delta_x} illustrates, a design $\xi^*$ supported on a regular simplex is optimal provided that $\SX$ is small enough and $p$ is large enough to make $\delta_{\xi^*}(x)$ concave at each $x_i$ (for symmetry reasons, we only need to check concavity at one vertex).
In fact, $p>2$ is sufficient. Indeed, assuming that $p>2$ and that $\psi(\cdot)$ is twice differentiable everywhere, with second-order derivative $\psi''(\cdot)$, except possibly at zero, direct calculation gives
\bea
\frac{d^2\delta_{\xi^*}(x)}{dxdx\TT}\bigg{|}_{x=x_1}= \frac{1}{d+1}\, \frac{p\,\psi^{p-1}(\Delta)\psi''(\Delta)}{\Delta^2} \sum_{i=2}^{d+1} (x_1-x_i)(x_1-x_i)\TT \,,
\eea
which is negative-definite (since $\psi''(\Delta)<0$, $\psi(\cdot)$ being maximal at $\Delta$). The right panel of Figure~\ref{F:delta_x} gives an illustration. Note that $p<2$ on the left panel, and the $x_i$ correspond to local minimas of $\delta_{\xi^*}(\cdot)$.
Figure~\ref{F:delta_x_2} shows a plot of $\delta_{\xi^*}(x)$ for $p=2$ and $K_0$ and $K_1$ being the Mat\'ern 3/2 and Mat\'ern 5/2 kernels $K_{0,1}$ and $K_{1,1.07}$, respectively, suggesting that the form of optimal designs may be in general quite complicated.

\begin{figure}[ht!]
\begin{center}
\includegraphics[width=.49\linewidth]{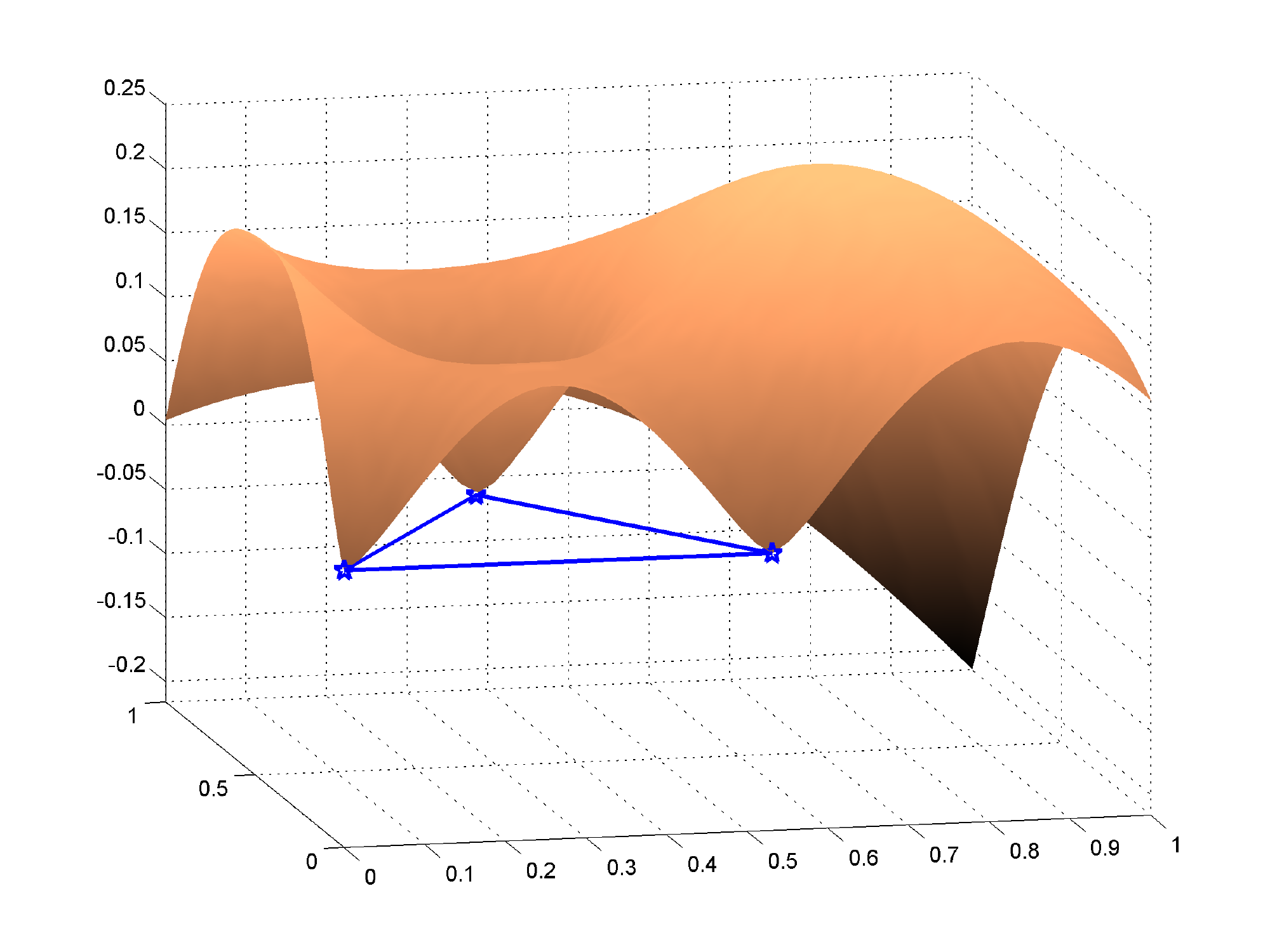}
\includegraphics[width=.49\linewidth]{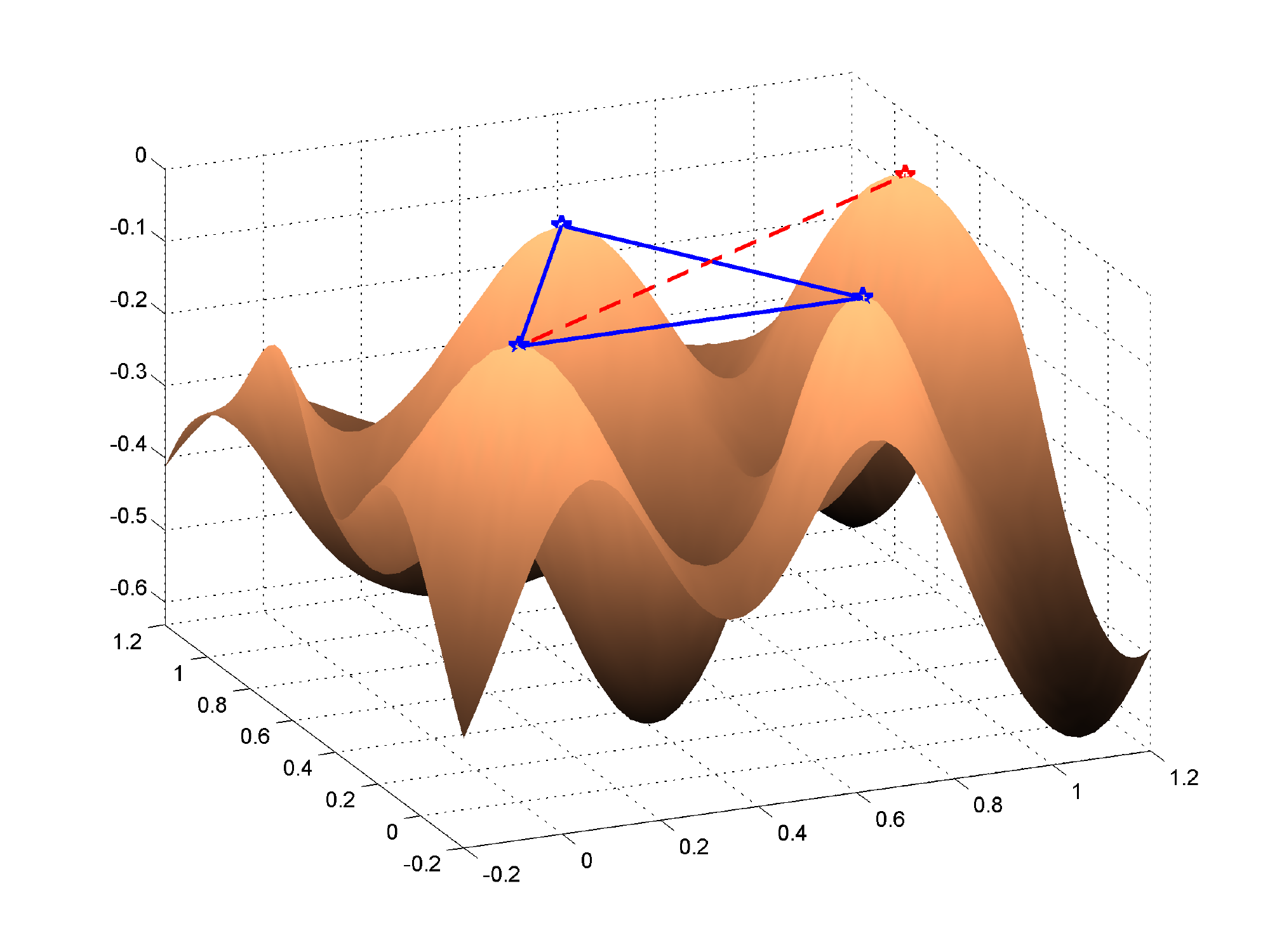}
\end{center}
\caption{\small Surface plot of $\delta_{\xi^*}(x)$ ($x\in\mathds{R}^2$), the support of $\xi^*$ corresponds to the vertices of the equilateral triangle in blue. Left: $K_0(x,x')=\exp(-\|x-x'\|)$ and $K_1=K_{1,1}$ ($\Delta\simeq 0.53$), $p=1.5$; Right: $K_0=K_{0,1}$, $K_1=K_{1,1}$ ($\Delta\simeq 0.7$), $p=10$; the red point $x_*$ is the symmetric of the origin $(0,0)$ with respect to the opposite side of the triangle.}\label{F:delta_x}
\end{figure}

\begin{figure}[ht!]
\begin{center}
\includegraphics[width=.49\linewidth]{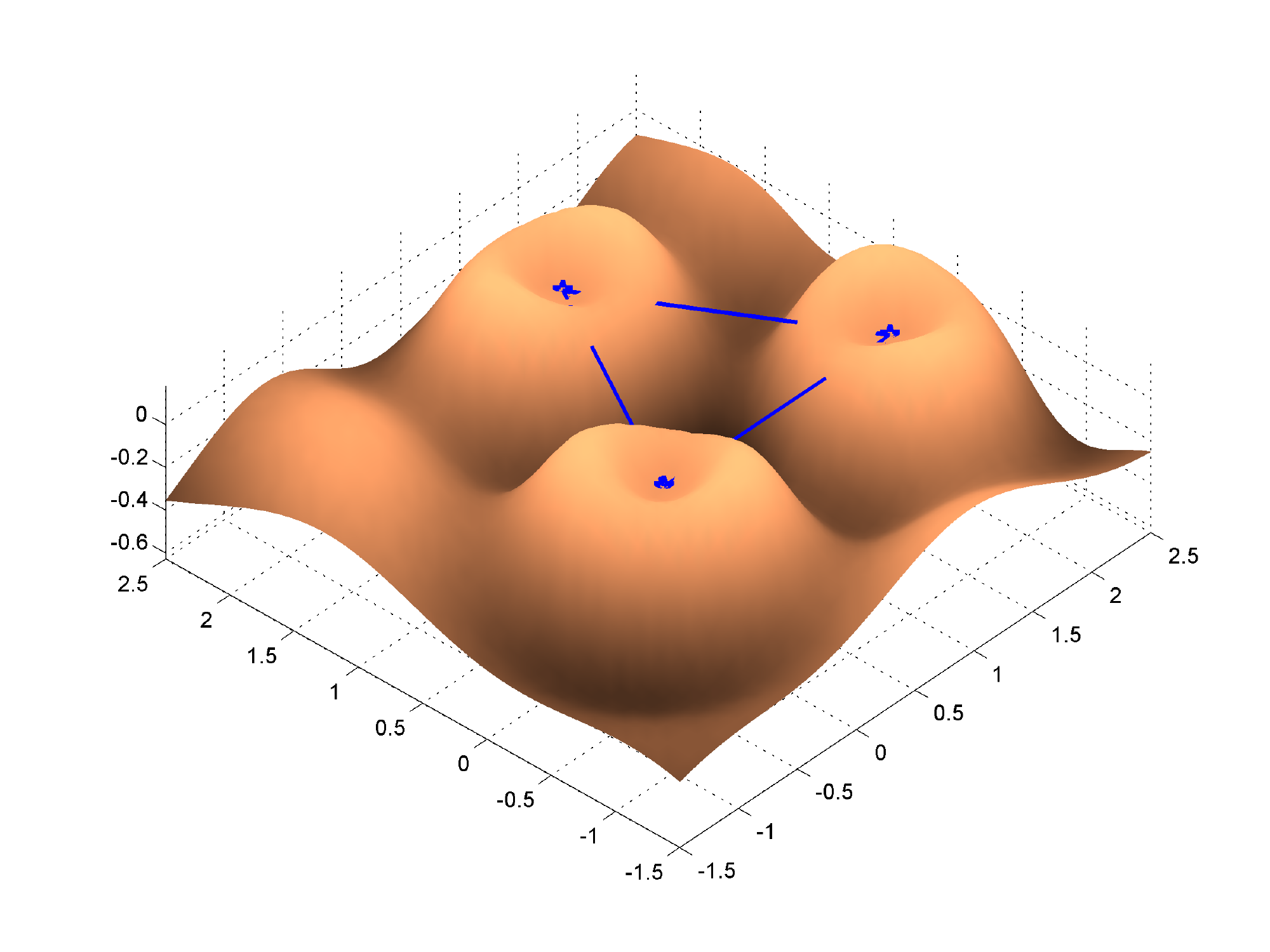}
\end{center}
\caption{\small Surface plot of $\delta_{\xi^*}(x)$ ($x\in\mathds{R}^2$), the support of $\xi^*$ corresponds to the vertices of the equilateral triangle in blue: $K_0=K_{0,1}$, $K_1=K_{1,1.07}$ ($\Delta\simeq 1.92$), $p=2$.}\label{F:delta_x_2}
\end{figure}

\section{A numerical example}\label{S:Examples}
\subsection{Exact designs}\label{Subsec:exact_designs}
In this section, we consider numerical evaluations of designs resulting from the prediction-based and distance-based criteria. {Here}, the rival models are {the isotropic versions of the covariance kernels used in Example 3 (Section \ref{sec:incr_uncond}) for the design space $\SX=[0,10]^2$, discretised at $n=25$ equally spaced points in each dimension.} For an agreement on the setting of correlation lengths in both kernels, we have applied a minimisation procedure. Specifically, {we have taken $\theta=\theta_0=1$ in $K_{0,\theta}(x,x')$ and adjusted the parameter in the second kernel minimising each of the distance-based criteria for the design $\Xb_{625}$ corresponding to the full grid. This resulted in $\theta_1= 1.0047$, $1.0285$, $1.0955$ and $1.3403$, respectively, for $\Phi_F,\Phi_1,\Phi_2$ and $\Phi_{KL}$.} We have finally chosen $\theta_1=1.07$, which seems to be compatible with the above values.

The left panel in Figure \ref{plot.maternkernels-Newparsets-2dim} shows the plot of the two Mat\'ern covariance functions at the assumed parameter values. This plot illustrates the similarity of the kernels which we aim to discriminate. The right panel in the figure refers to the plot of the absolute difference between the covariance kernels. The red line corresponds to the distance where the absolute difference between them is maximal. This is denoted by $\Delta$, which is equal to $\Delta=1.92$ in this case.

\begin{figure}[htp]
	\begin{center}
		\includegraphics[width=.45\linewidth]{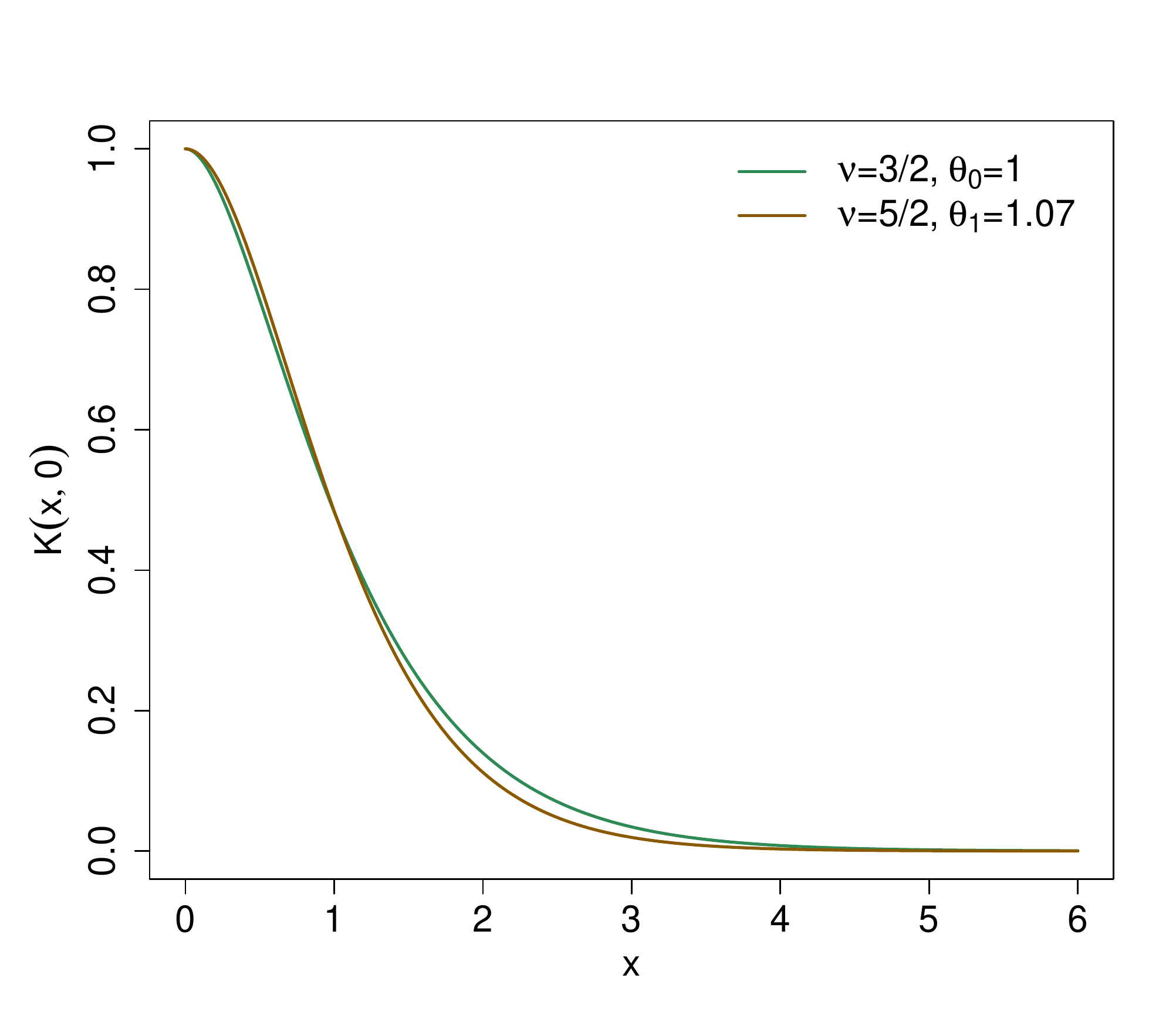}
		\includegraphics[width=.45\linewidth]{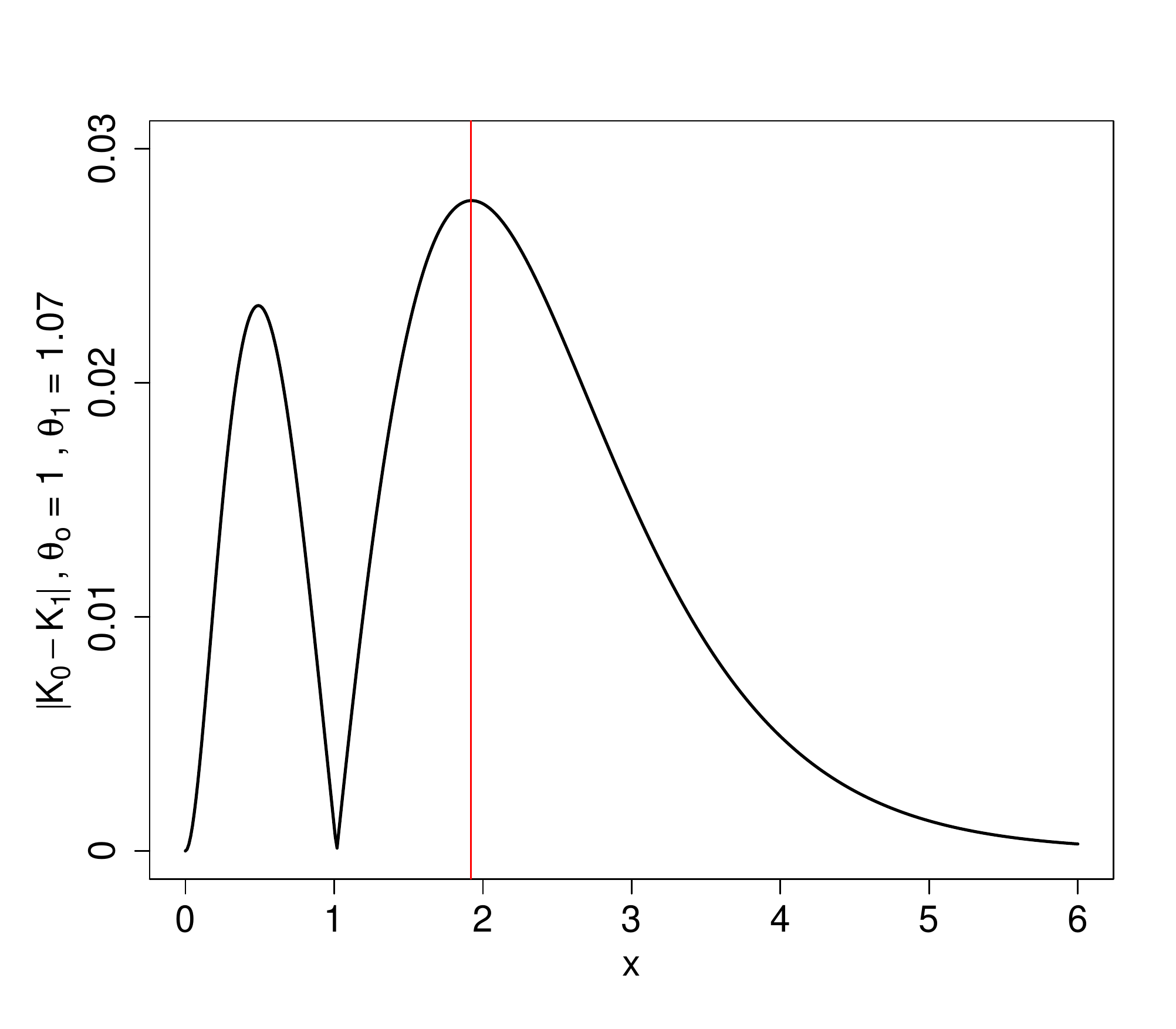}
		
	\end{center}
	\caption{{\footnotesize Left: Plot of the Mat\'ern covariance functions at the assumed parameter setting.
			Right: $\psi(t)=|K_{0,\theta_0}(t,0)-K_{0,\theta_1}(t,0)|, (\theta_0=1,\theta_1=1.07)$. 
	}}
	\label{plot.maternkernels-Newparsets-2dim}
\end{figure}

The sequential approach is the only case where the observations $\Yb_n$ corresponding to the previous design points $\Xb_n$ are used in the design construction. 
We use this information to estimate the parameter setting at each step. The (box)plots of the maximum likelihood (ML) estimates {$\hat{\theta}_0$ and $\hat{\theta}_1$ of the inverse correlation lengths $\theta_0$ and $\theta_1$} of $K_{0,\mt}(x,x')$ and $K_{1,\mt}(x,x')$, respectively, are presented in Figure \ref{Plot:MLest_theta}. This refers to the case where the first kernel, Mat\'ern 3/2, is the data generator. The $\hat{\theta}_0$ estimates converge to their null value, $\theta_0=1$, drawn as a red dashed line in the left panel of Figure \ref{Plot:MLest_theta}, {as expected due to the consistency of the ML estimator in this case}. For the second kernel to be similar to the first one (i.e., less smooth), the $\hat{\theta}_1$ estimates have increased (see the right panel). The {decrease} of the correlation length causes the covariance kernel to drop faster as a function of distance. We defer from presenting the opposite case (where the Mat\'ern 5/2 is the data generator), which is similar.

\begin{figure}[htp]	
	\begin{center}		
			\includegraphics[width=.45\linewidth]{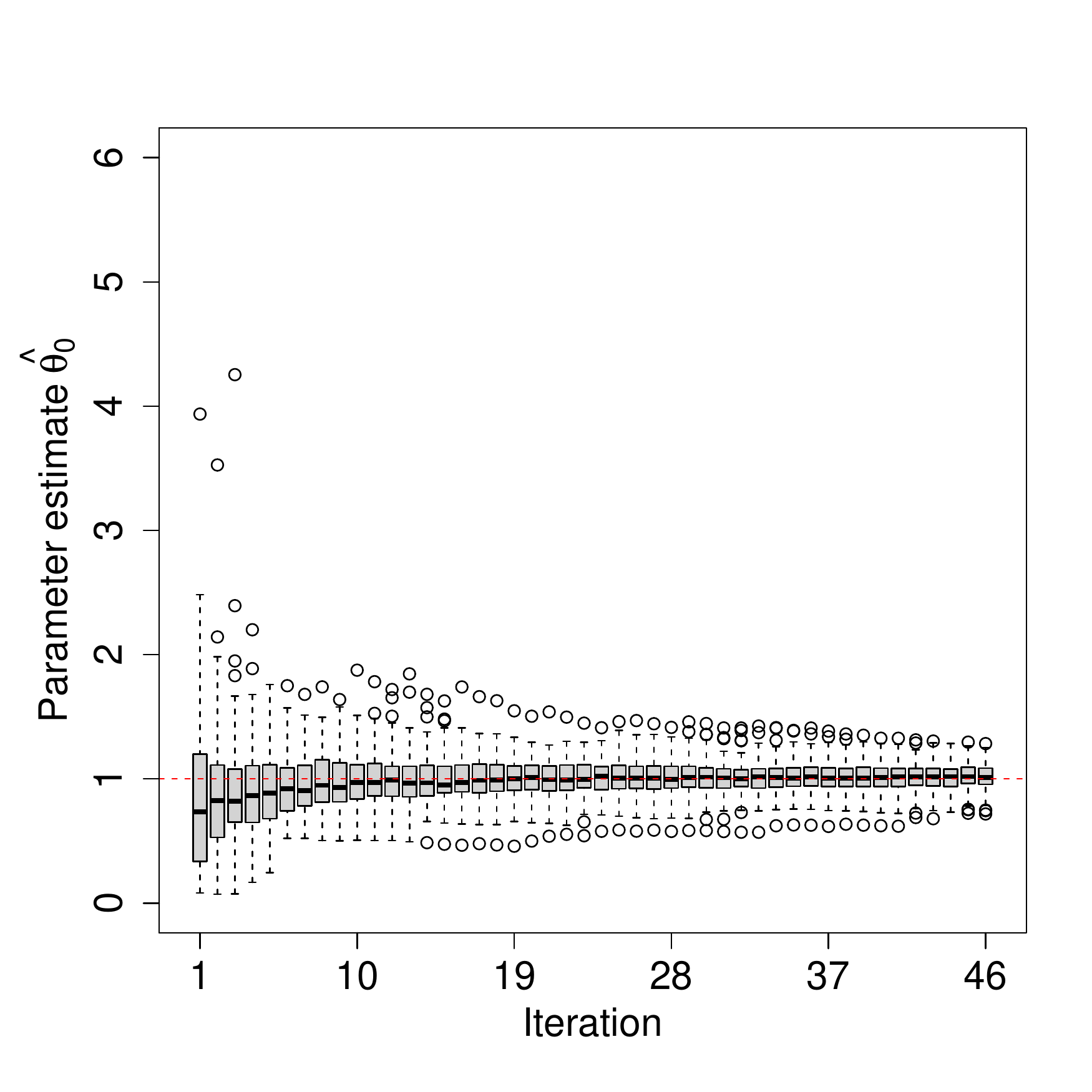}
			\includegraphics[width=.45\linewidth]{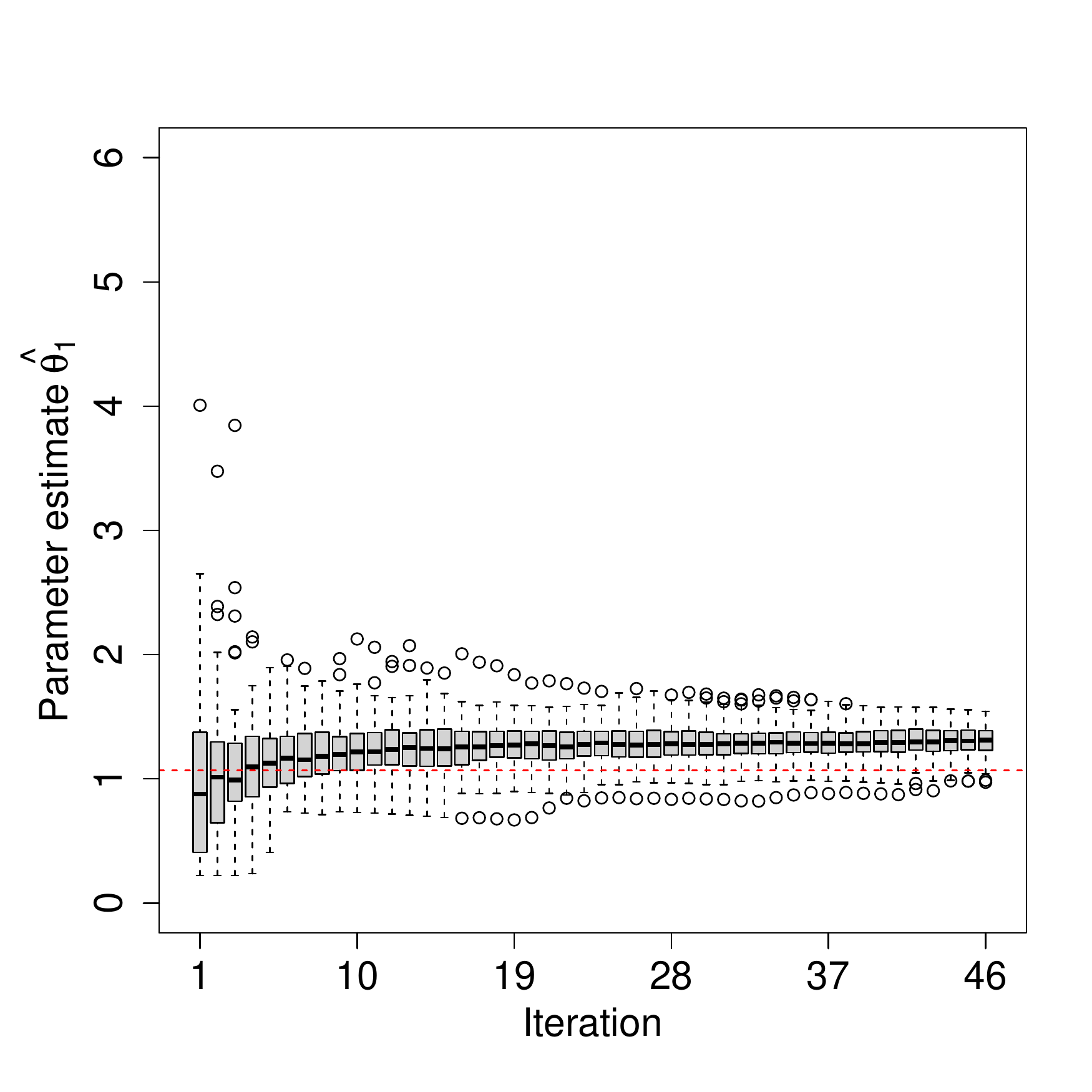}
	\end{center}
	\caption{Maximum likelihood estimates of the correlation lengths in Mat\'ern kernels. 
	}
	\label{Plot:MLest_theta}
\end{figure}

Apart from the {methods applied in Section \ref{S:distance-based-D}, we have considered some other static approaches for discrimination. $D_s$-optimal design is a natural candidate that} can be applied in the distance-based fashion. For $D_s$-optimality, we require the general form of the Mat\'ern covariance kernel, which is based on the modified Bessel function of the second kind (denoted by $C_{\nu}$). It is given by
\be \label{Besself_sec}
\Kb_{\nu}(r)=\dfrac{2^{1-\nu}}{\Gamma(\nu)}\left( \sqrt{2\nu} ~r\theta\right)^{\nu} C_{\nu}\left( \sqrt{2\nu} ~r\theta\right).
\ee
Smoothness, $\nu$, is considered as the parameter of interest, while the correlation length $\theta$ is assumed as nuisance. The first off-diagonal element in the $2\times 2$ information matrix, associated with the estimation of parameters $\thetav=(\theta,\nu)$, is
\be
M(\Xb_n,\thetav)_{12}=\dfrac{1}{2}\tr\left\lbrace \Kb_{\nu}^{-1}\dfrac{\partial \Kb_{\nu}}{\partial \theta} \Kb_{\nu}^{-1}\dfrac{\partial \Kb_{\nu}}{\partial \nu} \right\rbrace,
\ee
see, e.g., Eq.~(6.19) in \cite{muller2007collecting}. The other elements in the information matrix are calculated similarly. We have used the supplementary material of \cite{lee2018abc} to compute the partial derivatives of the Mat\'ern covariance kernel. Finally, the $D_s$-criterion is
\be \label{eq:D_scrit-spatial}
\Phi_{D_s }= |M(\Xb_n,\thetav)|/ |M(\Xb_n,\thetav)_{11}|,
\ee
where $M(\Xb_n,\thetav)_{11}$ is the element of the information matrix corresponding to the nuisance parameter (i.e., in $M(\Xb_n,\thetav)_{11}$ both partial derivatives are calculated with respect to $\theta$). 
In the examples to follow we consider local $D_s$-optimal design; that is, the parameters $\theta$ and $\nu$ are set at given values.

From a Bayesian perspective, models can be discriminated optimally when the difference between the expected entropies of the prior and the posterior model probabilities is maximised. This criterion underlies a famous sequential procedure put forward by \cite{box_discrimination_1967} and \cite{hill_note_1969}. Since such criteria typically cannot be computed analytically, several bounds were derived. The upper bound proposed by \cite{box_discrimination_1967} is equivalent to the symmetric Kullback-Leibler divergence $\Phi_{KL}$. \cite{hoffmann_numerical_2017} derives a lower bound based on a lower bound for the Kullback-Leibler divergence between a mixture of two normals, which is given by Eq.~\eqref{eq:BayesianLowerB_1} and is denoted by $\Phi_{\Gamma}$. Here, we assume equal prior probabilities. A more detailed account of Bayesian design criteria and their bounds is given in Appendix~\ref{sec:BayesianLowerUpperb}.

\begin{table}[htp]
	\caption{Comparison of average hit rates in different methods for the first numerical example.
}
	\label{table:com.all-hitrates-2dimMaterns}
	\hfill{\renewcommand{\arraystretch}{1.4}
		\begin{center}
			\resizebox{\textwidth}{!}{
				\begin{tabular}{lcccccccccc}
					\hline
					&\multicolumn{10}{c}{ Average hit rate }\\
					\cline{2-11}
					
					\text{\textbf{Design size}}&$5$&$ 6$&$ 7$&$8$&$9$&$10$&$20$&$30$&$40$&$ 50$\\
					\hline
					\text{\textbf{Sequential \eqref{SeqCond}}}&$0.500$&$ 0.535$&$ 0.540$&$ 0.595$&$ 0.570$&$ 0.640$&$ 0.695$&$ 0.715$&$ 0.740$&$ 0.770$\\
					\hline
					$\phi_A$&$ 0.505$&$ 0.500$&$ 0.530$&$ 0.525$&$ 0.505$&$ 0.510$&$ 0.520$&$ 0.535$&$ 0.585$&$ 0.635$\\
					\hline
					$\phi_B$&$0.520$&$ 0.545$&$ 0.575$&$ 0.585$&$ 0.615$&$ 0.650$&$ 0.785$&$ 0.875$&$ 0.900$&$ 0.910$\\
					\hline
					$\phi_{KL}$&$0.520$&$ 0.545$&$ 0.575$&$ 0.585$&$ 0.615$&$ 0.650$&$ 0.785$&$ 0.870$&$ 0.915$&$ 0.925$\\
					\hline
					$\Phi_F$&$0.580$&$\textbf{ 0.625}$&$ \textbf{0.620}$&$0.625$&$0.670$&$ \textbf{0.715}$&$0.795$&$ \textbf{0.900}$&$0.925$&$ 0.950$\\
					\hline
					$\Phi_1$&$0.525$&$ 0.520$&$ 0.555$&$ 0.540$&$ 0.550$&$ 0.610$&$ 0.725$&$ 0.890$&$ 0.910$&$ 0.920$\\
					\hline
					$\Phi_2$&$0.525$&$ 0.520$&$ 0.555$&$ 0.540$&$ 0.550$&$ 0.610$&$ 0.715$&$ 0.860$&$ 0.890$&$ 0.910$\\
					\hline
					$\Phi_{KL}$&$0.580$&$ \textbf{0.625}$&$ \textbf{0.620}$&$0.625$&$0.670$&$ \textbf{0.715}$&$ 0.795$&$ 0.895$&$0.925$&$ \textbf{0.955}$\\
					\hline
					$\Phi_{\Gamma}$&$\textbf{0.595}$&$ \textbf{0.625}$&$ 0.610$&$ \textbf{0.645}$&$ \textbf{0.675}$&$ 0.700$&$0.795$&$ 0.895$&$ \textbf{0.935}$&$ 0.940$\\
					\hline
					$\Phi_{D_s}$&$0.540$&$  0.575$&$  0.590$&$  0.620$&$  0.650$&$  0.675$&$  \textbf{0.805}$&$  0.850$&$  0.855$&$ 0.925$\\

				\end{tabular}
			}
		\end{center}
	}
\end{table}

Table \ref{table:com.all-hitrates-2dimMaterns} collects simulation results for the given example. We have included the sequential procedure \eqref{SeqCond} {as a benchmark} for orientation.
For all other {approaches} the true parameter values are used in the covariance kernels.
Concerning static (distance-based) designs based on maximisation of $\Phi_F,\Phi_1,\Phi_2,\Phi_{KL},\Phi_\Gamma,\Phi_{D_s}$, for each design size considered we first built a an incremental design and then used a classical exchange-type algorithm to improve it. These designs are thus not necessarily nested, i.e., $\Xb_n\not\subset\Xb_{n'}$ for $n<n'$.

Each design of size $n$ was then evaluated by generating $N=100$ independent sets of $n$ observations generated with the assumed true model, evaluating the likelihood functions for these sets of observations for both models, and then deciding for each set of observations which model has the higher likelihood value. The hit rate is the fraction of sets of observations where the assumed true model has the higher likelihood value. The procedure was repeated by assuming the other model to be the true one. The two hit rates are then averaged and stated in Table~\ref{table:com.all-hitrates-2dimMaterns}, which contains the results for all the criteria and design sizes we considered. For the special case of the sequential construction \eqref{SeqCond}, the design path depends on the observations generated at the previously selected design points; that is, unlike for the other criteria, for a given design size $n$ each random run produces a different design. To compute the hit rates for a particular $n$ we used $N=100$ independent runs of the experiment.

The hit rates reported in Table \ref{table:com.all-hitrates-2dimMaterns} reflect the discriminatory power of {the} corresponding designs. One can observe that $\Phi_F$ and as expected $\Phi_{KL}$ are outperforming in terms of hit rates. The Bayesian lower bound criterion $\Phi_{\Gamma}$ is similar to the symmetric $\Phi_{KL}$. The sequential design strategy (\ref{SeqCond}) does not behave as well as the outperforming ones. It is, however, the realistic scenario that one might consider in applications as it does not assume knowledge of the kernel parameters. The effect of this knowledge can thus be {partially} calibrated for by comparing the first line against the other criteria.

\subsection{Optimal design measure {for $\phi_p$}}\label{S:6.1.1}

Theorem \ref{Theorem:necessaryCon_Phip} also allows the use of approximate designs as it presents a necessary condition for optimality of the family of criteria $\phi_p,~p>0$. This is more extensively discussed in the previous section. Here we present the numerical results for two specific cases of $p=2$ and $p=10$.
To reach a design which might be numerically optimal (or at least nearly optimal), we have applied the Fedorov-Wynn algorithm \citep{fedorov1971design,wynn1970sequential} on a dense regular grid of candidate points.

Numerical results show that for very small $p$ (e.g., $p=1$) explicit optimal measures are hard to derive. The left panel in Figure \ref{plot.2dim-DisBased-Phi2,Phi10-measure} presents the measure $\xi_2^*$ obtained for $\phi_{2}$. To construct $\xi_2^*$, we have first calculated an optimal design on a dense grid by applying $1000$ iterations of the Fedorov-Wynn algorithm (see the comment following Theorem~\ref{Theorem:necessaryCon_Phip}); the design measure obtained is supported on 9 grid points. We then applied a continuous optimisation algorithm (library NLopt \citep{nlopt} through its R-interface {\tt nloptr}) initialised at this $9$-point design. The 9 support points of the resulting design measure $\xi_2^*$ are independent of the grid size; they receive unequal weights, proportional to the disk areas on Figure~\ref{plot.2dim-DisBased-Phi2,Phi10-measure}-left. Any translation or rotation of $\xi_2^*$ yields the same value of $\phi_2$.

As the order $p$ increases, we eventually reach an optimal measure with only three support points and equal weights. The right panel in Figure \ref{plot.2dim-DisBased-Phi2,Phi10-measure} corresponds to the optimal design measure computed for $\phi_{10}$. This has, similarly as before, resulted from application of a continuous optimisation initialised at an optimal 3-point design calculated with the Fedorov-Wynn algorithm on a grid. This optimal design measure $\xi_{10}^*$ has three support points, drawn as blue dots, with equal weights 1/3 represented by the areas of the red disks. The blue line segments between every two locations have length  $\Delta\simeq1.92$, reflecting the ideal interpoint distance (see the right panel of Figure \ref{plot.maternkernels-Newparsets-2dim}), in agreement with corresponding discussions in Section~\ref{odm}.
Also here the optimal designs are rotationally and translationally invariant, and thus any design of such type is optimal as long as the design region is large enough to fit it.

\begin{figure}[htp]
	\begin{center}
		\includegraphics[width=.45\linewidth]{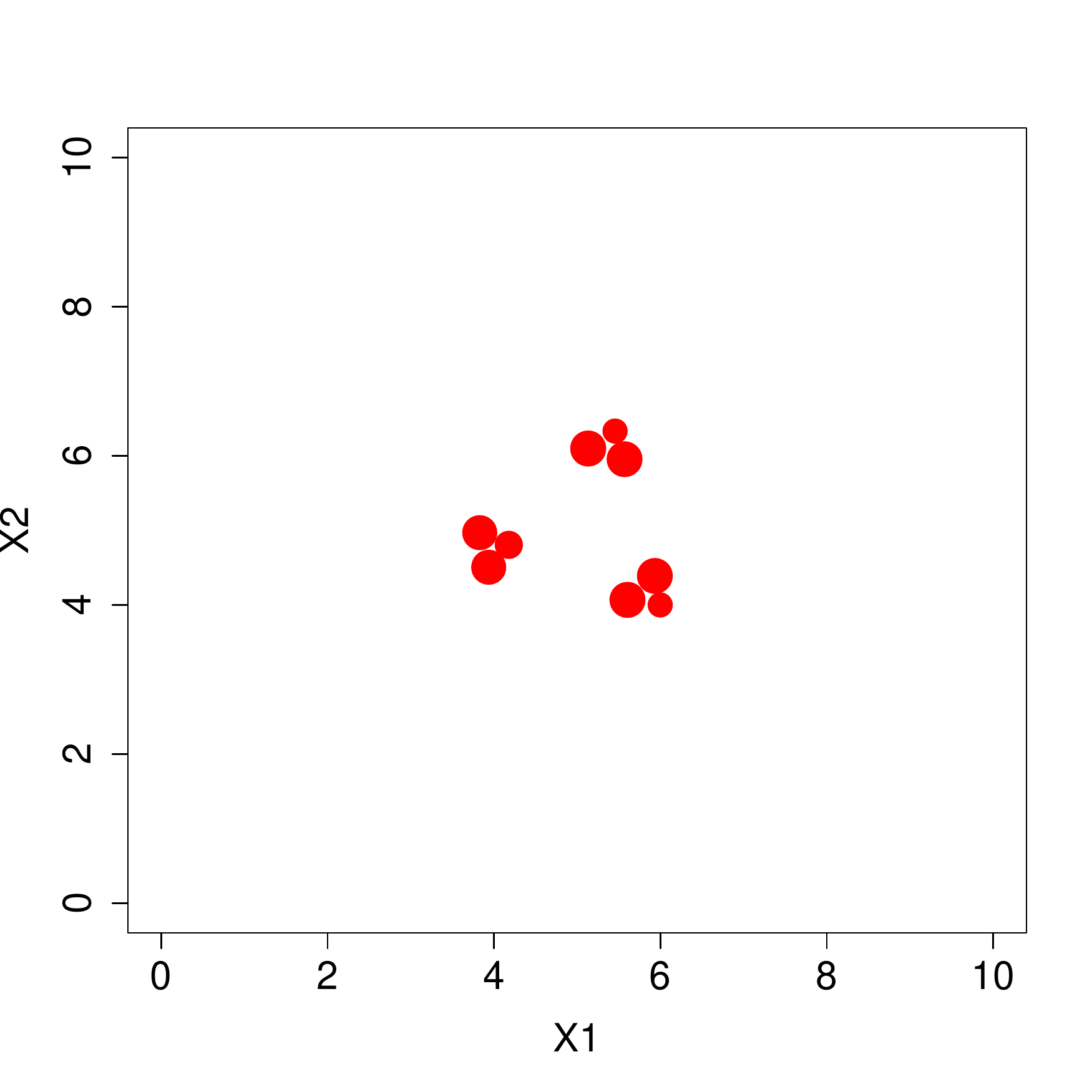}
		\includegraphics[width=.45\linewidth]{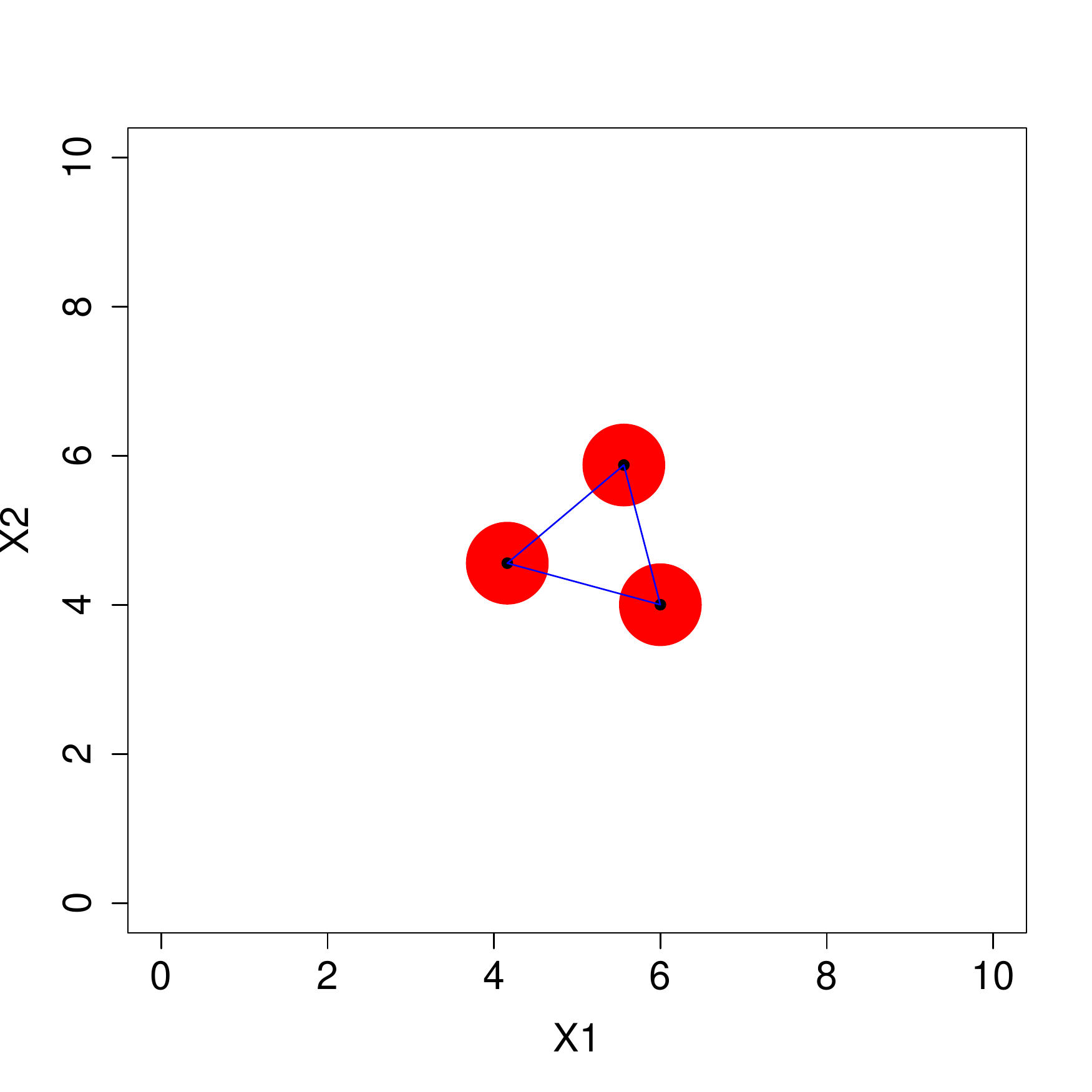}
	\end{center}
	\caption{{\footnotesize Left: The optimal measure for $\phi_{2}$. Right: The optimal measure for $\phi_{10}$.}}
	\label{plot.2dim-DisBased-Phi2,Phi10-measure}
\end{figure}

\section{Conclusions}

In this paper we have considered the design problem for the discrimination of Gaussian process regression models. This problem differs considerably from the well-treated one in standard regression models and thus offers a multitude of challenges. While the KL-divergence is a straightforward criterion, it comes with the price of being computationally demanding and lacking convenient simplifications such as design measures. We have therefore introduced a family of criteria that allow such a simplification at least in special cases and have investigated its properties. We have also compared the performance of these and other potential criteria on several examples and see that KL-divergence can be effectively replaced by simpler criteria without much loss in efficiency. In particular designs based on the Fr\'echet-distance between covariance kernels seem to be competitive. Results from the approximate design computations indicate that for classical isotropic kernels, designs with $d+1$ support points placed at the vertices of a simplex of suitable size are optimal for distance-based criteria $\phi_p$ when $p$ is large enough.

\section*{Acknowledgments}
This work was partly supported by project INDEX (INcremental Design of EXperiments) ANR-18-CE91-0007 of the French National Research Agency (ANR) and  I3903-N32 of the Austrian Science Fund (FWF).


\bibliographystyle{apalike}
\bibliography{Discrimination}

\appendix
\section*{Appendix}
\section{Notes on Box-Hill-Hunter Bayesian criteria for model discrimination between Gaussian random fields} \label{sec:BayesianLowerUpperb}

Chapter~5 of \cite{hoffmann_numerical_2017} contains an overview of Bayesian design criteria for model discrimination and some useful bounds on them.
We assume there are $M$ models $m_0,\ldots,m_{M-1}$. The most common Bayesian design criterion for model discrimination has the following form:
\begin{equation}
	\Phi_{\Lambda}(\Xb_k) = -\sum_{i=0}^{M-1} p(m_i) \log (p(m_i))  + \int_{\Yb_k \in \mathcal{Y}} p(\Yb_k) \sum_{i=0}^{M-1} p(m_i | \Yb_k) \log(p(m_i|\Yb_k)) \: \dd \Yb_k, \label{eq:Lambda1}
\end{equation}
where the data $\Yb_k = (Y_1(x_1),\ldots,Y_k(x_k))\TT$ are observed at the design $\Xb_k = (x_1,\ldots,x_k)$, $p(m_i)$ denotes the prior and $p(m_i|\Yb_k)$ the posterior model probability of model $m_i$ and $p(\Yb_k)$ is the marginal distribution of $\Yb_k$ with respect to the models. 
Hence, this criterion is the (expected) difference of the model entropy and the conditional model entropy (conditional on the observations). The posterior model probability $p(m_i|\Yb_k)$ is defined by
\begin{equation*}
	p(m_i|\Yb_k) \propto p(\Yb_k|m_i) p(m_i),
\end{equation*}
where $p(\Yb_k|m_i)$ is the likelihood of model $m_i$ (marginalised over the parameters), and $p(\Yb_k)$ is given by
\begin{equation*}
 	p(\Yb_k) = \sum_{i=0}^{M-1}  p(\Yb_k|m_i) p(m_i).
\end{equation*}
The first term in \eqref{eq:Lambda1} does not depend on the design and can therefore be ignored.

A common alternative formulation of criterion~\eqref{eq:Lambda1} is the one adopted by \cite{box_discrimination_1967} and \cite{hill_note_1969}, which will henceforth be called Box-Hill-Hunter (BHH)
criterion:
\begin{equation}
\Phi_{\Lambda}(\Xb_k) = \sum_{i=0}^{M-1} p(m_i) \int_{\Yb_k \in \mathcal{Y}}  p(\Yb_k|m_i) \log\left( \frac{p(\Yb_k|m_i)}{p(\Yb_k)} \right) \: \dd \Yb_k. \label{eq:Lambda2}
\end{equation}

In our case, if we assume point priors for the kernel parameters, we have
\begin{equation*}
	p(\Yb_k|m_i) = \varphi(\Yb_k | \etab_{k,i}, \Kb_{k,i}),
\end{equation*}
where $\etab_{k,i} = (\eta_{1,i}(x_1),\ldots,\eta_{k,i}(x_k))\TT$ is the mean vector of model $i$ at design $\Xb_k$, $\Kb_{k,i}$ is the $k \times k$ kernel matrix of model $i$ with elements given by $\{\Kb_{k,i}\}_{j,l} = K_i(x_j, x_l)$, and $\varphi(\cdot|\etab,\Kb)$ is the normal pdf with mean vector $\etab$ and variance-covariance matrix $\Kb$.

For example, for a static design involving $n$ design points, we set $k = n$ and assume that $\etab_{n,i} = \mathbf{0}$ for each design $\Xb_n$. The model probabilities $p(m_i)$ would just be the prior model probabilities before having collected any observations.

In a sequential design setting, where $n$ observations $\Yb_n$ have already been observed at locations $\Xb_n$ and we want to find the optimal design point $x$ where to collect our next observation, we have $k=1$ and set $\etab_{k,i}$ to the conditional mean $\widehat{\eta}_{n,i}(x) = \kb_{n,i}(x)\TT \Kb_{n,i}^{-1} \Yb_n$ and $\Kb_{k,i}$ to the conditional variance $\rho_{n,i}^2(x) = K_i(x,x) - \kb_{n,i}(x)\TT \Kb_{n,i}^{-1} \, \kb_{n,i}(x)$, where $\kb_{n,i}(x)\TT = (K_i(x,x_1),\ldots,K_i(x,x_n))$, see Section~\ref{S:3.1}. The prior model probabilities would have to be set to the posterior model probabilities given the already observed data:
\begin{equation*}
	p(m_i) = p(m_i|\Yb_n) \propto \varphi(\Yb_n | \, \mathbf{0}, \Kb_{n,i}) \: p(m_i).
\end{equation*}

It follows that $p(\Yb_k)$ is a mixture of normal distributions. The criterion representations~\eqref{eq:Lambda1} and \eqref{eq:Lambda2} cannot be computed directly. 
However, several bounds have been developed for the criterion, the most famous being the classic upper bound derived by \cite{box_discrimination_1967}.

\subsection{Upper bound}

The upper bound has the following form (see also \citet[Thm.~5.2, p.~168]{hoffmann_numerical_2017}):
\begin{equation*}
	\Phi_{U}(\Xb_k) = \frac{1}{2} \sum_{i=0}^{M-1} \sum_{j=0}^{M-1} p(m_i) p(m_j) \left\{ \left\|\etab_{k,i}-\etab_{k,j}\right\|^2_{\Kb_{k,j}^{-1}} + \tr\left(\Kb_{k,i} \Kb_{k,j}^{-1}\right) - n \right\}.
\end{equation*}

For $M=2$, the formula simplifies to
\begin{eqnarray*} \label{eq:BayesianUpperB}
\Phi_{U}(\Xb_k) & = & \frac{1}{2}  p(m_0) p(m_1) \biggl\{ \left\|\etab_{k,0}-\etab_{k,1}\right\|^2_{\Kb_{k,0}^{-1}} + \left\|\etab_{k,0}-\etab_{k,1}\right\|^2_{\Kb_{k,1}^{-1}} \\
& & \qquad \qquad \qquad + \tr\left(\Kb_{k,0} \Kb_{k,1}^{-1}\right) + \tr\left(\Kb_{k,1} \Kb_{k,0}^{-1}\right) - 2n \biggr\}.
\end{eqnarray*}

This is equivalent to the symmetric Kullback-Leibler divergence that we use as the criterion {$\Phi_{KL}$  (with $p(m_0)=p(m_1)=1/2$ and $\etab_{k,0}=\etab_{k,1}=\0b$)} .

\subsection{Lower bound}

\citet[Sec.~7]{hershey_approximating_2007} derive a lower bound for the Kullback-Leibler divergence between a mixture of two normals, see also \citet[Thm.~5.4 and Cor.~5.5, pp.~173--174]{hoffmann_numerical_2017}. This result is then used by \cite{hoffmann_numerical_2017} to find a lower bound for the BHH criterion $\Phi_{\Lambda}(\Xb_k)$ \citep[Thm.~5.9, p.~178]{hoffmann_numerical_2017}.
This lower bound is given by
\begin{equation*}
	\Phi_{\Gamma}(\Xb_k) = -\sum_{i=0}^{M-1}  p(m_i) \log \left\{ \sum_{j=0}^{M-1} p(m_j) \exp\left(-\frac{1}{2} \boldsymbol{\Gamma}(\Xb_k)_{ij} \right) \right\},
\end{equation*}
where
\begin{equation*}
	\boldsymbol{\Gamma}(\Xb_k)_{ij} = \left\|\etab_{k,i}-\etab_{k,j}\right\|^2_{\Kb_{k,j}^{-1}} + \tr\left(\Kb_{k,i} \Kb_{k,j}^{-1}\right) - \log \det \left(\Kb_{k,i} \Kb_{k,j}^{-1}\right) - n.
\end{equation*}

For $M=2$, as is the relevant case for our setup we get
\begin{eqnarray} \label{eq:BayesianLowerB_1}
	\Phi_{\Gamma}(\Xb_k) & = & -p(m_0) \log \biggl\{ p(m_0) \nonumber \\
	& & \qquad \qquad \quad + p(m_1)  \exp\biggl( -\frac{1}{2} \bigl[ \left\|\etab_{k,0}-\etab_{k,1}\right\|^2_{\Kb_{k,1}^{-1}} + \tr\left(\Kb_{k,0} \Kb_{k,1}^{-1}\right) \nonumber \\
	& & \qquad \qquad \qquad \qquad \qquad \qquad - \log \det\left(\Kb_{k,0} \Kb_{k,1}^{-1}\right) -n \bigr] \biggr) \biggr\} \nonumber \\
	& & - p(m_1) \log \biggl\{ p(m_1) \nonumber\\
	& & \qquad \qquad \quad + p(m_0)  \exp\biggl( -\frac{1}{2} \bigl[ \left\|\etab_{k,0}-\etab_{k,1}\right\|^2_{\Kb_{k,0}^{-1}} + \tr\left(\Kb_{k,1} \Kb_{k,0}^{-1}\right) \nonumber \\
	& & \qquad \qquad \qquad \qquad \qquad \qquad - \log \det\left(\Kb_{k,1} \Kb_{k,0}^{-1}\right) -n \bigr] \biggr) \biggr\} 
\end{eqnarray}
where $\varphi_i(\cdot) = \varphi(\cdot|\etab_{k,i},\Kb_{k,i})$,
which we are also using to compute designs in Section~\ref{Subsec:exact_designs} (again with $p(m_0)=p(m_1)=1/2$ and $\etab_{k,0}=\etab_{k,1}=\0b$).

\end{document}